\documentclass[%
a4paper, UKenglish, cleveref, nameinlink,  
colorlinks=true, linkcolor=blue!75!green, urlcolor=blue!75!green, citecolor=blue!75!green, autoref, thm-restate]{lipics-v2021}

\usepackage{apptools,thmtools} 
\let\origrestatable=\restatable
\renewcommand{\restatable}[1][]{%
  \origrestatable[%
    \ifstrempty{#1}%
      {$\star$}%
      {\IfAppendix{\hyperref[#1]{$\star$}}{\hyperref[#1*]{$\star$}}}%
  ]%
}

\graphicspath{{./lipics/}}%

\title{Upward-Planar Drawings with Bounded Span}

\author{Patrizio Angelini}{John Cabot University, Italy}{pangelini@johncabot.edu}{https://orcid.org/0000-0002-7602-1524}{}
\author{Sabine Cornelsen}{University of Konstanz, Germany}{sabine.cornelsen@uni-konstanz.de}{https://orcid.org/0000-0002-1688-394X}{}
\author{Giordano Da Lozzo}{Roma III, Italy}{giordano.dalozzo@gmail.com}{https://orcid.org/0000-0003-2396-5174}{}
\author{Fabrizio Frati}{Roma III, Italy}{frati@dia.uniroma3.it}{https://orcid.org/0000-0001-5987-8713}{}
\author{Philipp Kindermann}{Trier University, Germany}{kindermann@uni-trier.de}{https://orcid.org/0000-0001-5764-7719}{}
\author{Ignaz Rutter}{University of Passau, Germany}{rutter@fim.uni-passau.de}{https://orcid.org/0000-0002-3794-4406}{}
\author{Johannes Zink}{TU Munich, Germany}{johannes.zink@tum.de}{https://orcid.org/0000-0002-7398-718X}{}

\authorrunning{P. Angelini, S. Cornelsen, G. Da Lozzo, F. Frati, P. Kindermann, I. Rutter, J. Zink}

\Copyright{Patrizio Angelini, Sabine Cornelsen, Giordano Da Lozzo, Fabrizio Frati, Philipp Kindermann, Ignaz Rutter, and Johannes Zink} %

\Crefname{observation}{Observation}{Observations}
\Crefname{algorithm}{Algorithm}{Algorithms}
\Crefname{section}{Sect.}{Sects.}
\Crefname{observation}{Observation}{Observations}
\Crefname{lemma}{Lemma}{Lemmas}
\Crefname{claim}{Claim}{Claims}
\Crefname{figure}{Fig.}{Figs.}
\Crefname{figure}{Fig.}{Figs.}
\Crefname{invariant}{Inv.}{Invs.}
\Crefname{enumi}{Condition}{Conditions}
\Crefname{property}{Property}{Properties}
\Crefname{assumption}{Assumption}{Assumptions}



\category{} %

\relatedversion{} %

\acknowledgements{This work was initiated at the Annual Workshop on Graph and Network Visualization (GNV~2025), Chania, June 2025.}

\hideLIPIcs  
\nolinenumbers %

\DeclareTextFontCommand{\emph}{\color{blue!75!green}\em}

\usepackage{caption}
\usepackage{subcaption}
\usepackage{dsfont}
\usepackage{todonotes}

\newcommand{\spn}{\ensuremath{\mathrm{span}}}
\newcommand{\vc}{\ensuremath{\mathrm{vc}}}
\newcommand{\vi}{\ensuremath{\mathrm{vi}}}
\newcommand{\td}{\ensuremath{\mathrm{td}}}
\newcommand{\att}{\ensuremath{\mathrm{att}}}

\begin{document}

\maketitle

\begin{abstract}
   We consider upward-planar layered drawings of directed graphs, i.e., crossing-free drawings in which each edge is drawn as a y-monotone curve going upward from its tail to its head, and the y-coordinates of the vertices are integers. The span of an edge in such a drawing is the absolute difference between the y-coordinates of its endpoints, and the span of the drawing is the maximum span of any edge. The span of an upward-planar graph is the minimum span over all its upward-planar drawings.

   We study the problem of determining the span of upward-planar graphs and provide both combinatorial and algorithmic results. On the combinatorial side, we present upper and lower bounds for the span of directed trees. On the algorithmic side, we show that the problem of determining the span of an upward-planar graph is NP-complete already for directed trees and for biconnected single-source graphs. Moreover, we give efficient algorithms for several graph families with a bounded number of sources, including $st$-planar graphs and graphs where the planar or upward-planar embedding is prescribed. Furthermore, we show that the problem is fixed-parameter tractable with respect to the vertex cover number and the treedepth plus the span.
\end{abstract}

\section{Introduction}

The visualization of directed graphs is a central topic in Graph Drawing and has sparkled extensive research efforts spanning several decades. A natural objective to effectively convey the hierarchical information is to construct crossing-free visualizations in which edges ``flow upward'' according to their direction. This has been formalized within the two classic notions of upward-planar drawings~\cite{dibattista/tamassia:wg87,BertolazziBMT98,DBLP:journals/comgeo/LozzoBF20,DBLP:journals/tcs/Frati24,GargTamassia1995,DBLP:journals/siamcomp/HuttonL96,JansenKKLMS23} and layered drawings~\cite{angelini_etal:tcs15,DBLP:journals/algorithmica/BannisterDDEW19,DBLP:journals/jocg/BlazejKKS0024,DBLP:journals/tcs/BrucknerR25,DBLP:journals/tsmc/BattistaN88,JungerLM98,DBLP:journals/talg/KlemzR19}.
An \emph{upward-planar (layered)} drawing of a directed acyclic graph, for short DAG, maps each vertex to a point in the plane (with integer y-coordinates) and each edge to a strictly y-monotone Jordan arc from its tail to its head, so that no two edges intersect except at common endpoints. By \cite{DBLP:journals/algorithmica/EadesFLN06,pach/toth:04}, we may assume that the edges are drawn as straight-line segments. A DAG is \emph{upward-planar} if it has an  upward-planar drawing. 

The \emph{span} of an edge $e$ in an upward-planar layered drawing~$\Gamma$ is the difference between the y-coordinate of its head and the y-coordinate of its tail. It is denoted by~$\spn_{\Gamma}(e)$ or simply by~$\spn(e)$, if~$\Gamma$ is clear from the context.  The \emph{span} of drawing~$\Gamma$, denoted by~$\spn(\Gamma)$, is the maximum span of any of its edges. The \emph{span}~$\spn(G)$ of an upward-planar graph~$G$ is the minimum span of any of its upward-planar layered drawings. See \cref{fig:intro-example} for an example of a DAG with span~three.  In this paper, we consider the problem $k$-\textsc{span upward planarity}, which takes as input an upward-planar graph and asks whether it has span at most $k$. Furthermore, if the graph comes with a prescribed (upward-)planar embedding, then such an embedding must be preserved in the produced drawing.

\begin{figure}[t]
    \centering
    \includegraphics[page=2]{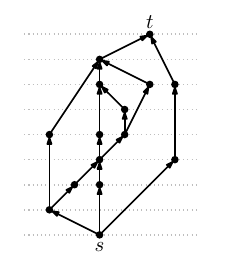}
    \caption{A directed graph with span three. The drawing has span three and there is a directed $s$-$t$-path of length 8 and a directed $s$-$t$-path of length 3. So the span has to be at least $8/3 > 2$.}
    \label{fig:intro-example}
\end{figure}

\subparagraph*{Motivations and related work.} The problem of constructing layered drawings with bounded span has previously been studied for undirected graphs~\cite{bekos_etal_span:gd24, GiacomoDLMMW25}. The authors study the problem of computing acyclic orientations of the edges of a graph such that the resulting DAG has span at most $k$. 
A significant motivation to study drawings with bounded span lies in their connection with the \emph{edge-length ratio}, that is, the maximum ratio between the Euclidean lengths of any two edges in the drawing~\cite{BlazjFL21,BorrazzoF20,GiacomoDLMMW25}. 
This relationship, in the directed case, is: If an upward-planar layered drawing with span $k$ exists, then an upward-planar layered drawing with edge-length ratio at most $k(1+\epsilon)$ exists, for any $\epsilon>0$, see~\cite{GiacomoDLMMW25}.

Undirected graphs that admit leveled-planar drawings with span 1 are also known as
\emph{(proper) leveled planar}.
Heath and Rosenberg showed that recognizing graphs that admit such drawings is
NP-complete~\cite{heath/rosenberg:jc92}. On the contrary, we observe that deciding whether an upward-planar DAG admits an upward-planar layered drawing of span $1$, i.e.,  solving $1$-\textsc{span upward planarity},  can be done in linear time, as all layers are fixed and thus the problem reduces to level-planarity testing~\cite{JungerLM98}.
If the number of allowed distinct y-coordinates (layers) is bounded by a constant, then the problem can be solved in polynomial time for both directed and undirected graphs~\cite{dujmovic_etal:algorithmica08}.

Drawing edges with small span is also a goal in the layer-assignment step of the celebrated Sugiyama framework~\cite{sugiyama_etal:81} for drawing directed graphs. There, however, edges are allowed to cross.
Minimizing the sum of the span of all edges in the Sugiyama framework can be done in polynomial time using an ILP-approach with a totally unimodular constraint matrix~\cite{gansner_etal:93}. This approach can be extended to decide in polynomial time whether there is a (not necessarily planar) upward layered drawing in which the maximum span is at most~$k$.
In this case, the formulation is simpler, as it does not require a global objective.

\subparagraph*{Our Results.}%
We first provide upper and lower bounds for the span of directed trees in terms of the number $n$ of vertices, the maximum degree $d$, and the maximum length $\ell$ of a directed path. Paths and caterpillars have span 1 (\cref{th:path}) 
unless an upward-planar embedding is given; in this case, even paths require a span that is linear in $n$ (\cref{prop:Pn}). If no planar embedding is given and each vertex is a source or a sink, i.e., $\ell=1$, 
then the span of directed trees is exactly $\lfloor d/2 \rfloor + 1$ (\cref{th:ub-degree-ss,cor:Tk}). In this case, the span is also lower bounded by $\log_2n/(2\log_2\log_2 n)$ 
(\cref{cor:Tk_boundn}).  If there are longer directed paths, 
there exist binary directed trees with span at least $(\ell+1)/2$ and $(n-1)/6$ (\cref{prop:binaryTree}).  In general, the span for directed trees is in $\mathcal O((d-1)^{\ell}) \cap \mathcal O(\ell \cdot n^{\log_2\varphi})$ with $\varphi=(1+\sqrt{5})/2$ (\cref{th:ub-path-degree,th:ub-longest-path}) and lower bounded by $(d-1)^\ell/2\ell$ (\cref{cor:Tkell}).

The $k$-\textsc{span upward planarity} problem is NP-complete, even for biconnected single-source graphs (\cref{thm:nph-single-source}) and for directed trees (\cref{thm:nph-trees}).
On the positive side, 
we show that 
it 
can be solved in linear time for planar st-graphs (\cref{thm:st-planar}) and single-source outerplanar graphs (\cref{th:outerplanar}), %
and in XP time with respect to the number of sources for planar DAGs with a prescribed planar embedding (\cref{cor:xp}). Finally, we show that $k$-\textsc{span upward planarity} is in FPT
if parameterized by the vertex cover number (\cref{thm:fpt-vc}), or by the span plus the %
treedepth (\cref{thm:fpt-para+span}).

\subparagraph*{Preliminaries.} 
For a directed edge $(u,v)$, we call $u$ its \emph{tail} and $v$ its \emph{head}. A vertex is a \emph{source} (\emph{sink}) if it is not the head (tail) of any edge. The \emph{underlying graph} of a directed graph~$G$ is the undirected graph obtained from $G$ by ignoring the edge directions. The \emph{degree} of a vertex is the number of its incident edges.
The \emph{degree} of a graph is the maximum degree of any of its vertices. We denote by $y(v)$ the y-coordinate of a vertex~$v$ in a drawing.

A \emph{planar embedding} of a connected planar graph fixes the outer face and the cyclic order of the edges around each vertex. 
A \emph{plane graph} is a planar graph equipped with a planar embedding. 
In an upward-planar drawing of a DAG (a) the cyclic order of the edges around each vertex is such that  all outgoing edges are consecutive (a planar embedding satisfying this condition is called \emph{bimodal}) and (b) exactly one angle between consecutive edges around each sink or source is greater than $\pi$ (this is called the \emph{large angle} of the source or sink). A \emph{source-switch} (\emph{sink-switch}) for a face $f$ in an upward-planar drawing of a DAG $G$ is a source (resp.\ sink) for the subgraph of $G$ composed of the vertices and edges on the boundary of $f$. An \emph{upward-planar embedding} is a bimodal embedding together with an assignment of large angles to sources and sinks such that each face $f$ has a number of large angles equal to the number of source-switches minus $1$, if $f$ is an internal face, or plus $1$, if $f$ is the outer face. An \emph{upward-plane graph} is a directed graph equipped with an upward-planar embedding. 

An \emph{st-graph} is a directed graph with a single source $s$ and a single sink $t$. An st-graph is \emph{planar} if it admits an \emph{st-planar embedding}, i.e., a planar embedding in which $s$ and $t$ are incident to the outer face. A planar st-graph equipped with an st-planar embedding is a \emph{plane st-graph}. An st-planar embedding uniquely determines an upward-planar embedding (the only large angles are those incident to $s$ and $t$ in the outer face)~\cite{BattistaT88}.  %

We use $[k]$ as an abbreviation for $\{1, 2, ..., k\}$.

\section{Lower Bounds}\label{sec:lowerBounds}

We prove lower bounds for the span of some graph classes.
The span of a graph with $n$ vertices is at most $n-1$. 
In fact, in the worst case, the span is in $\Theta(n)$ for paths with a fixed upward-planar embedding, even if there is no directed subpath of length two.

\begin{proposition}\label{prop:Pn}
    Let $P_n=\langle v_1,\dots,v_n\rangle$ be a path that is oriented such that $v_i$ is a source if~$i$ is odd and a sink if $i$ is even (i.e. $P_n$ contains no directed path of length two).  
    If we fix the upward-planar embedding of $P_n$ so that all large angles are on the same side of $P_n$ (see \cref{fig:spiralT} for an illustration of $P_{11}$), 
    then $P_n$ has span at least $n/2$. 
\end{proposition}
    \begin{figure}[h]
    	\centering
    \begin{minipage}[t]{0.33\textwidth}
		\centering
		\includegraphics[page=19]{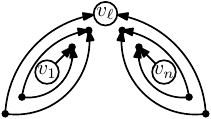}
		\subcaption{path $P_{11}$, fixed embedding}
        \label{fig:spiralT}
	\end{minipage}%
	\begin{minipage}[t]{0.33\textwidth}
		\centering
		\includegraphics[page=18]{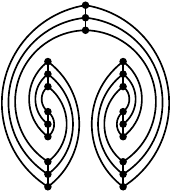}
		\subcaption{3-connected graph $G_7$}
        \label{fig:spiralG}
	\end{minipage}%
	\begin{minipage}[t]{0.33\textwidth}
		\centering
		\includegraphics[page=6]{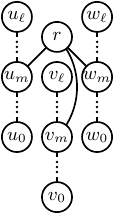}
		\subcaption{binary tree $T$}
        \label{fig:binaryTree}
	\end{minipage}
    		\caption{Graphs with large span.
            }
    	\label{fig:spiral}
    \end{figure}
\begin{proof}
    In an upward-planar drawing of $P_n$ with the given embedding, path $P_n$
    can be partitioned into two subpaths $\langle v_1,\dots,v_\ell \rangle$ and $\langle v_\ell, \dots , v_n\rangle$ with $1 \leq \ell \leq n$ that are so-called {\em coils}   %
    with the following property: For any $i$ and $j$ with $1 \leq j < i \leq \ell$ or $\ell \leq i < j \leq n$, if $v_i$ is a source then $y(v_i) < y(v_j)$ and if $v_i$ is a sink then $y(v_i) > y(v_j)$; see \cite[Lemma~2]{frati:08}. This implies  that the edge between $v_{\ell-1}$ and $v_\ell$ has span $\ell-1$ and the edge between $v_\ell$ and $v_{\ell+1}$ has span $n - \ell$ (if these edges exist). So the span is at least $\max(\ell-1,n-\ell) \geq n/2$.  
\end{proof}

By combining three copies of $P_n$ (see \cref{fig:spiralG}), we can also consider variable embeddings: 
\begin{corollary}
\label{cor:3-connected}
    For each $n$ there is a 3-connected upward-planar graph with $3n$ vertices, 
    degree at most four, maximum length of a directed path at most five, and span at least $n/2$.
\end{corollary}

\begin{proof}
    We consider three copies of the path from \cref{prop:Pn} with alternating sources and sinks: $\langle u_1,\dots,u_n\rangle$, $\langle v_1,\dots,v_n\rangle$, $\langle w_1,\dots,w_n\rangle$. For $1 \leq i \leq n$ odd, we add the edges $(u_i,v_i)$ and $(v_i,w_i)$ and for  $1 \leq i \leq n$ odd, we add the edges $(w_i,v_i)$ and $(v_i,u_i)$. Finally, we add the edge $(u_1,w_1)$ and the edge $(u_n,w_n)$ or $(w_n,u_n)$, depending on whether $n$ is odd or even. See \cref{fig:spiralG} for an illustration. The resulting graph $G_n$ is upward-planar, has degree at most four, and no directed path has more than five edges.  Since $G_n$ is 3-connected, it has a unique planar embedding up to a flip and the choice of the outer face. 
    Let $P_n$ be the middle path $\langle v_1,\dots,v_n\rangle$. Since each vertex in $P_n$ is the head and the tail of some edge not in $P_n$, the upward-planar embedding of $P_n$ within $G_n$ is fixed; more precisely, the large angles are all on the same side of $P_n$. Thus, by \cref{prop:Pn}, $P_n$ has span at least $n/2$. Thus, also $G_n$ has span at least $n/2$.
\end{proof}

Next, we consider trees without fixed embedding. 
If the length of a directed path is unbounded then even the span of a binary tree can be linear in the number of its vertices.

\begin{proposition}\label{prop:binaryTree}
    For each $\ell$ and each $n\geq n_0 =3\ell+4$, there is an $n$-vertex directed tree, with degree at most three, maximum length of a directed path at most $\ell$, and span at least $\lceil (\ell+1)/2 \rceil = \lceil (n_0-1)/6 \rceil$.
\end{proposition}
\begin{proof}
    Let $T$ be any directed tree %
    that contains three disjoint directed paths $\langle u_0,\dots,u_\ell\rangle$, $\langle v_0,\dots,v_\ell\rangle$, $\langle w_0,\dots,w_\ell\rangle$ and an additional vertex $r$ %
    with three edges $(u_m,r)$, $(v_m,r)$, and $(w_m,r)$ with $m=\lceil \ell/2 \rceil$. See \cref{fig:binaryTree}. 
    Consider an upward-planar drawing of $T$. W.l.o.g.\ let the incoming edges of $r$ be $(u_m,r)$, $(v_m,r)$, and $(w_m,r)$ from left to right. 
    If the edge $(v_m,r)$ had span smaller than $\lceil (\ell+1)/2 \rceil$ then $u_0$ and $w_0$ are below $v_m$, so the $\lceil (\ell+1)/2 \rceil$ vertices $v_{m},\dots,v_\ell$ would have to be drawn on distinct layers between the paths $\langle u_0,\dots,u_m,r\rangle$ and $\langle w_0,\dots,w_m,r\rangle$. But there are less than $\lceil (\ell+1)/2 \rceil$ layers. 
\end{proof}

We continue with trees with arbitrary degrees, but bounded directed path length.
We recursively define for each positive $d \in \mathds Z$ an undirected tree~$T'_d$ rooted at a vertex~$r$ of degree~$d$. $T'_1$ is a single edge and $T'_2$ is a path of length two.
If $d>2$ then the subtree rooted at each neighbor of~$r$ is $T'_{d-2}$ if $d$ is odd and $T'_{d-3}$ otherwise. See \cref{fig:Tk}. Let ${T_d}$ be an orientation of $T'_d$ in which all vertices are sources or sinks.

\begin{figure}[h]
		    \centering
	\begin{minipage}[t]{0.15\textwidth}
		\centering
		\includegraphics[page=1]{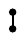}
		\subcaption{$T'_1$}
	\end{minipage}
	\begin{minipage}[t]{0.15\textwidth}
		\centering
		\includegraphics[page=2]{figures/highDegreeHighSpan.pdf}
		\subcaption{$T'_2$}
	\end{minipage}
	\begin{minipage}[t]{0.15\textwidth}
		\centering
		\includegraphics[page=3]{figures/highDegreeHighSpan.pdf}
		\subcaption{$T'_3$}
	\end{minipage}
	\begin{minipage}[t]{0.15\textwidth}
		\centering
		\includegraphics[page=4]{figures/highDegreeHighSpan.pdf}
		\subcaption{$T'_4$}
	\end{minipage}
		\begin{minipage}[t]{0.15\textwidth}
		\centering
		\includegraphics[page=21]{figures/highDegreeHighSpan.pdf}
		\subcaption{$T_4$, $r$ sink}
	\end{minipage}
	\begin{minipage}[t]{0.15\textwidth}
		\centering
		\includegraphics[page=22]{figures/highDegreeHighSpan.pdf}
		\subcaption{$T_4$, $r$ source}
	\end{minipage}\\[1em]
\begin{minipage}[t]{0.4\textwidth}
		\centering
		\includegraphics[page=7]{figures/highDegreeHighSpan.pdf}
		\subcaption{$T'_d$, with $d$ odd}
	\end{minipage}
	\begin{minipage}[t]{0.4\textwidth}
	\centering
	\includegraphics[page=6]{figures/highDegreeHighSpan.pdf}
	\subcaption{$T'_d$, with $d$ even}
\end{minipage}
	\caption{$T'_d$ has maximum degree $d$ and span at least $\lceil d/2 \rceil$ if all vertices are sources or sinks.}
	\label{fig:Tk}
\end{figure}

\begin{lemma}
\label{lem:Tk}
	The tree ${T_d}$ has maximum degree $d$. If $d \geq 7$, then the number of vertices of $T_d$ is greater than $2^d$ and less than  $d!$. The span of $T_d$ is at least $\lceil d/2 \rceil$.  
	Moreover, if $d$ is even then let $v_1,\dots,v_{d/2},$ $v'_{d/2},\dots,v'_1$ be the neighbors of the root in the cyclic order in which they appear in some upward-planar layered drawing of $ T_d$. Then the span of the edge between the root and $v_i$ and $v_i'$, respectively, is at least $i$ for $i=1,\dots, d/2$.
\end{lemma}

\begin{proof}
    Obviously, the degree of $T_d$ is at most $d$. Let $n_d$ be the number of vertices of $T_d$. We show by induction on $d$ that $2^d < n_d < d!$ for $d \geq 5$ odd and for $d \geq 8$ even. For $d=5$, we have $n_5 = 5 \cdot 7 + 1 = 36 < 120 = 5!$ and $36 > 32 = 2^5$. For $d>5$ odd we get $n_d = d\cdot n_{d-2} + 1 < d(d-2)! + 1 = \frac{d!}{d-1}+1 = d! - \frac{d!(d-2)-d+1}{d-1}<d!$ and $n_d > d \cdot 2^{d-2} + 1 > 2^d$. For $d\geq 8$ even, we get $n_d = d\cdot n_{d-3} + 1 < d \cdot (d-3)! + 1 < d!$ and $n_d > d \cdot 2^{d-3} + 1 > 2^d$. 
    
	For the span, we first consider the case that $d$ is even. We use induction on $d$.
    The trees~$T_1$ and~$T_2$ have span at least one. Assume now that~$d>2$ is even and consider an upward-planar drawing of~${T_d}$ where the neighbors of the root~$r$ appear in the order  $v_1,\dots,v_{d/2},v'_{d/2},\dots,v'_1$. Consider first a subtree of~${T_d}$ isomorphic to~$T_{d-2}$ rooted at~$r$ and containing the children  $v_1,\dots,v_{d/2-1},v'_{d/2-1},\dots,v'_1$ of~$r$. See the black vertices in \cref{fig:firstCase}. By the inductive hypothesis, the edges between~$r$ and~$v_i$ and~$v_i'$ have span at least $i$ for $i=1,\dots,d/2-1$. 
	
	Assume now that no edge has span at least~$d/2$. Consider a tree isomorphic to~$T_{d-2}$ rooted at~$v_{d/2}$; it contains the subtree of~${T_d}$ rooted at~$v_{d/2}$, the root~$r$ of~${T_d}$ and parts of~$d-5$ subtrees of~${T_d}$ rooted at children of~$r$ other than~$v_{d/2}$. See the black vertices in \cref{fig:secondCase}. By the inductive hypothesis applied to this~$T_{d-2}$ there must be at least two edges with span at least~$d/2 - 1$ incident to~$v_{d/2}$, one of which might be the edge~$e$ between~$r$ and~$v_{d/2}$. But if the edge~$e$ had span exactly~$d/2 - 1$ then there was no space to also draw the second edge~$e'$ of span~$d/2 - 1$, a contradiction. See the dashed edge in \cref{fig:thirdCase}. Recall that the edges between~$r$ and~$v_{d/2-1}$ and~$v'_{d/2-1}$, respectively, had span at least~$d/2-1$. Analogously, we conclude that the edge between~$r$ and~$v'_{d/2}$ has span at least~$d/2$. 
    
    The proof for odd~$d$ is analogous. Let $v_1,\dots,v_{\frac{d-1}2},v_{\frac{d+1}2},v'_{\frac{d-1}2},\dots,v'_1$ be the children of~$r$. The subgraph of~${T_{d}}$ without the subtree rooted at~$v_{\frac{d+1}2} $ contains a~${T_{d-1}}$ rooted at~$r$ and so the span of the edges between~$r$ and~$v_i$ or~$v_i'$ is at least~$i$ for $i=1,\dots,\frac{d-1}2$. Further, $T_{d}$ contains a subtree isomorphic to~$T_{d-1}$ rooted at~$v_{\frac{d+1}2}$. Thus, by the same argument as above, the edge between~$v_{\frac{d+1}2}$ and~$r$ has span $(d-1)/2 + 1= \lceil d/2 \rceil$.
\end{proof}

\begin{figure}[h]
		\centering
        \begin{subfigure}[t]{0.3\textwidth}
			\centering
			\includegraphics[page=9]{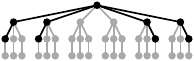}
			\subcaption{\label{fig:firstCase}}
		\end{subfigure}
        \begin{subfigure}[t]{0.3\textwidth}
			\centering
			\includegraphics[page=11]{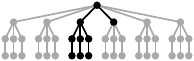}
			\subcaption{\label{fig:secondCase}}
		\end{subfigure}
        \begin{subfigure}[t]{0.35\textwidth}
			\centering
			\includegraphics[page=13]{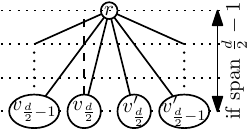}
			\subcaption{\label{fig:thirdCase}}
		\end{subfigure}
		\caption{Illustration of the proof \cref{lem:Tk}. 
        The drawing in (a) and (b) is not upward. }
		\label{fig:proofTk}		
	\end{figure}

\begin{corollary}\label{cor:Tk_boundn}
    For infinitely many $n$ there is a tree with $n$ vertices and longest directed path length one that has span greater than $\log_2 n / (2 \log_2 \log_2 n)$.  
\end{corollary}
\begin{proof}
    Let $d \geq 7$ and let $n_d$ be the number of vertices of $T_d$. Then $\log_2 n_d < \log_2 d! < \log_2 d^d = d \log_2 d < d \log_2 \log_2 n_d$.
    Thus, the span of $T_d$ is $d/2 > \log_2 n_d / (2 \log_2 \log_2 n_d)$.
\end{proof}

For $d$ even, we can improve the lower bound of \cref{prop:Pn} by one. In \cref{sec:upperBounds}, we will see that this bound is tight.

\begin{theorem}
\label{cor:Tk}
		For each~$d>2$ there is a tree where each vertex has degree at most~$d$, each directed path has length at most one, and the span is at least~\mbox{$\lfloor d/2 \rfloor + 1$}. 
\end{theorem}

\begin{proof}
	The statement follows directly from \cref{lem:Tk}, if~$d$ is odd. If~$d$ is even, consider the tree~$T$ shown in \cref{fig:proofCorTk} for the case~$d=4$. 
There is a central sink~$r$ of degree~$d$. Each neighbor of~$r$ is the  root of a subtree containing a~${T_d}$. Among the neighbors of~$r$ there are three  where two more neighbors are again a root of a~${T_d}$.  The maximum degree of~$T$ is~$d$. Assume~$T$ has an upward-planar layered drawing with span~$d/2$. Assume without loss of generality that~$r$ has y-coordinate~0. Let $v_1,\dots v_{d/2},v'_{d/2},\dots,v'_1$ be the neighbors~$r$ in the cyclic order in which they appear around~$r$. 

By \cref{lem:Tk}, $v_i$ and~$v_i'$, respectively, have y-coordinate at most~$-i$ for $i=1,\dots,d/2$. By our assumption, the span of the edge between~$r$ and~$v_{d/2}$ and~$v'_{d/2}$ is at most~$d/2$, i.e., their y-coordinate is exactly~$-d/2$. Together with a (second) edge incident to $v_{i+1}$ ($v'_{i+1}$) of span at least~$i+1$ this implies that~$v_i$ and~$v_i'$, respectively have y-coordinate exactly~$-i$ for $i=1,\dots,d/2-1$. Since among $v_1,\dots v_{d/2},v'_{d/2},\dots,v'_1$ there are three that are the root of a~${T_d}$ there are $1 \leq i < j \leq d/2$ such that~$v_i$ and~$v_j$ or~$v_i'$ and~$v_j'$ are both the root of a~${T'_d}$. We assume the first. Let~$u$ and~$v$ be the neighbors of~$v_i$ other than~$r$ that are roots of a~$T'_d$. Then~$u$ and~$v$ have y-coordinate at most~$d/2-i$. Observe that both,~$v_{i+1}$ and~$v'_1$ are incident to an edge~$e$ and~$e'$, respectively, of span~$d/2$ and its end vertices have y-coordinate  $d/2-i-1$ and~$d/2-1$, respectively. We distinguish two cases. If $u$ or~$v$, say~$u$, has y-coordinate $y < d/2-i$ then~$u$ is between the intersection point of $e$ and~$e'$ with the horizontal line with y-coordinate~$y$. Thus, the second edge with span~$d/2$ incident to~$u$ cannot be drawn without a crossing. If both, $u$ and~$v$ have y-coordinate~$d/2-i$, assume that $u$, $v$, and the intersection point of~$e'$ with layer~$d/2-i$ is in this order. Then the second edge with span~$d/2$ incident to~$v$ cannot be drawn without a crossing. 	
\end{proof}

		\begin{figure}[h]
		\centering
        \begin{minipage}[t]{\textwidth}
        \centering
        \includegraphics[page=2]{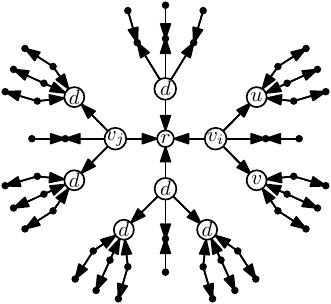}
        \end{minipage}\\[1em]
		\begin{minipage}[t]{\textwidth}
			\centering
			\includegraphics[page=12]{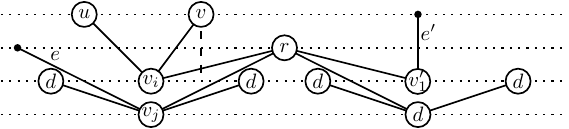}
		\end{minipage}%
		\caption{Tree~$T$ used in the proof of \cref{cor:Tk} in the case~$d=4$. Each vertex with some label is the root of a~$T'_d$ contained in~$T$.
        }
		\label{fig:proofCorTk}		
	\end{figure}

Finally, we consider also trees with longer directed paths.

\begin{corollary}\label{cor:Tkell}
	For each~$d>2$
    and each~$\ell$ there is a tree where each vertex has degree at most~$d$, each directed path has length at most~$\ell$, and the span is at least~\mbox{$(d-1)^\ell/2\ell$}.
\end{corollary}
\begin{proof}
	Consider the tree~${T_{(d-1)^{\ell}}}$. For each vertex~$v$ of~$T_{(d-1)^\ell}$, replace~$v$ and the edges to $v$'s children by a perfect $(d-1)$-ary tree~$T_v$  of height~$\ell$. All edges are directed from the root~$v$ to the leaves if~$v$ was a source and from the leaves to~$v$ if~$v$ was a sink. 
    Merge the children of~$v$ with distinct leaves of~$T_v$. 
    An upward-planar drawing~$\Gamma$ with span~$s$ of the thus constructed tree $T$ induces an upward-planar drawing~$\Gamma'$ of~${T_{(d-1)^\ell}}$ in which the vertices of~${T_{(d-1)^\ell}}$ are on the same layers as in~$\Gamma$.
    The span of~$\Gamma'$ is at most~$s\cdot \ell$ by construction, but at least~\mbox{$(d-1)^\ell / 2$} by \cref{lem:Tk}. Thus, $s \geq (d-1)^\ell/2\ell$.	
\end{proof}

\section{Upper Bounds}\label{sec:upperBounds}

In this section, we show upper bounds on the span of $n$-vertex directed trees. By \cref{prop:binaryTree}, such a span might be $\Omega(n)$; hence, our upper bounds apply to families of directed trees or rely on additional parameters beyond $n$. We begin with the following observation which, when compared with \cref{prop:Pn}, shows that the possibility to choose an upward-planar embedding plays an important role in the construction of drawings with bounded span.

\begin{observation} \label{th:path}
The span of any directed tree whose underlying graph is a path or a caterpillar is $1$.
\end{observation}

\begin{proof}
We draw a directed tree whose underlying graph is a path as an $x$-monotone curve that goes up and down between consecutive layers, according to the edge directions. For a directed caterpillar, we draw its spine as above and then embed the leaves connected to a spine vertex~$v$ on the \mbox{layers above and below $v$, according to the direction of their incident edges.} 
\end{proof}

We next consider bounded-degree directed trees in which the longest directed path has length $1$. These are the trees that we use in \cref{lem:Tk} in order to prove lower bounds on the span. We get the following upper bound, which by \cref{cor:Tk} is tight.

\begin{theorem}
\label{th:ub-degree-ss}
Let~$T$ be a directed tree in which the longest directed path has length~$1$, and let~$d$ be the degree of~$T$. Then~$T$ has span at most~$\lfloor \frac{d}{2}\rfloor +1$.
\end{theorem}

\begin{proof}
We can assume that~$d\geq 3$, as if~$d=2$ then~$T$ is a path, hence it has span~$1$ by \cref{th:path}, and the theorem follows. We root~$T$ at any leaf~$r$. This ensures that every vertex has at most~$d-1$ children. Since the longest directed path in~$T$ has length~$1$, every vertex of~$T$ is either a source or a sink. Assume that~$r$ is a source, the case in which it is a sink is symmetric. Let~$s$ be any source of~$T$ and let~$S$ be the subtree of~$T$ rooted at~$s$. We describe a recursive construction, depicted in \cref{fig:boundedDegreepathlengthOne-app},  of an upward-planar drawing~$\Gamma$ of~$S$ satisfying the following properties:

    \begin{figure}[h]
	\centering
	\includegraphics[page=6, scale =1.2]{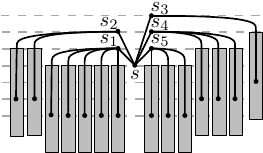}
	\caption{Illustration for the proof of \cref{th:ub-degree-ss}. Here, $s$ has $k=5$ children. %
        }
	\label{fig:boundedDegreepathlengthOne-app}
\end{figure}

\begin{itemize}
	\item~$s$ is \emph{top-accessible}, i.e., it is possible to draw a~$y$-monotone line that starts at~$s$, that cuts the top-side of the bounding box of~$\Gamma$, and that does not intersect~$\Gamma$ other than at~$s$;
	\item the \emph{upper height} of~$\Gamma$, which is the number of layers that intersect~$\Gamma$ and that are not below the layer of~$s$, is at most~$\lfloor \frac{d}{2}\rfloor +1$; note that the layer of~$s$ counts towards the value of the upper height; and 
	\item the span of every edge of~$S$ in~$\Gamma$ is at most~$\lfloor \frac{d}{2}\rfloor +1$. 
\end{itemize}   
When~$s=r$ and~$S=T$, this gives us the desired drawing of~$T$. Let~$n$ be the number of vertices of~$S$. If~$n=1$, then~$\Gamma$ is trivially constructed. Otherwise, let~$s_1,\dots,s_k$ be the children of~$s$ and note that such vertices are all sinks. Also, for~$i=1,\dots,k$, let~$t_{i,1},\dots,t_{i,f(i)}$ be the children of~$s_i$ and note that such vertices are all sources. Recursively construct a drawing~$\Gamma_{i,j}$ of each subtree~$T_{i,j}$ rooted at a vertex~$t_{i,j}$. Embed~$s$ at any layer. Then embed the~$\lfloor \frac{k}{2}\rfloor$ children~$s_1,\dots,s_{\lfloor \frac{k}{2}\rfloor}$ of~$s$ to the left of~$s$ on the~$\lfloor \frac{k}{2}\rfloor$ layers right above~$s$, one child per layer, with bottom-to-top order~$s_1,\dots,s_{\lfloor \frac{k}{2}\rfloor}$; embed the remaining~$\lceil \frac{k}{2}\rceil$ children~$s_{\lfloor \frac{k}{2}\rfloor+1},\dots,s_k$ of~$s$ to the right of~$s$ on the~$\lceil \frac{k}{2}\rceil$ layers right above~$s$, one child per layer, with top-to-bottom order~$s_{\lfloor \frac{k}{2}\rfloor+1},\dots,s_k$. Draw the edges incident to~$s$ as~$y$-monotone curves in such a way that their clockwise order in~$\Gamma$ is~$(s,s_1),\dots,(s,s_k)$. Place the drawings~$\Gamma_{i,j}$ with~$i=1,\dots,\lfloor \frac{k}{2}\rfloor$ and~$j=1,\dots,f(i)$ side-by-side to the left of~$s$, so that~$\Gamma_{i,j}$ is to the right of~$\Gamma_{i',j'}$ whenever~$i<i'$, and so that the top side of the bounding box of~$\Gamma_{i,j}$ is one layer lower than the layer of~$s_i$. Symmetrically, place the drawings~$\Gamma_{i,j}$ with~$i=\lfloor \frac{k}{2}\rfloor+1,\dots,k$ and~$j=1,\dots,f(i)$ side-by-side to the right of~$s$, so that~$\Gamma_{i,j}$ is to the left of~$\Gamma_{i',j'}$ whenever~$i<i'$, and so that the top side of the bounding box of~$\Gamma_{i,j}$ is one layer lower than the layer of~$s_i$. Connecting each vertex~$s_i$ to its children via~$y$-monotone curves completes the construction of~$\Gamma$. 

It is easy to see that~$\Gamma$ is a top-accessible upward-planar drawing of~$S$. The upper height of~$\Gamma$ is equal to~$1$, for the layer of~$s$, plus the number of children of~$s$ that are placed to the right of~$s$, hence it is equal to~$1+ k-\lfloor \frac{k}{2}\rfloor=\lceil \frac{k}{2}\rceil+1\leq \lceil \frac{d-1}{2}\rceil+1=\lfloor \frac{d}{2}\rfloor +1$. The span of each edge in a subtree~$T_{i,j}$ is at most~$\lfloor \frac{d}{2}\rfloor +1$ by induction. Also, the span of any edge incident to~$s$ is at most the upper height of~$\Gamma$ minus one, hence at most~$\lfloor \frac{d}{2}\rfloor$. Finally, the span of any edge~$(t_{i,j},s_i)$ is smaller than or equal to the upper height of~$\Gamma_{i,j}$; indeed, the edge~$(t_{i,j},s_i)$ has span at most the upper height of~$\Gamma_{i,j}$ minus one from~$t_{i,j}$ up to the top side of the bounding box of~$\Gamma_{i,j}$, and then span one from the top side of the bounding box of~$\Gamma_{i,j}$ up to~$s_i$. Hence, the span of~$(t_{i,j},s_i)$ is at most~$\lfloor \frac{d}{2}\rfloor +1$. This completes the induction, and hence the proof of the theorem.
\end{proof}

Next, we study trees in which the length of the longest path has a fixed upper bound~$\ell$.

\begin{theorem}
\label{th:ub-path-degree}
Let~$T$ be a directed tree, let~$d$ be the degree of~$T$, and let~$\ell$ be the length of the longest directed path in~$T$. Then $T$ has span in~$O\left((d-1)^{\ell}\right)$.
\end{theorem}

\begin{proof}
We can assume that~$d\geq 3$, as if~$d=2$ then~$T$ is a path, hence it has span~$1$ by \cref{th:path}, and the theorem follows. We root~$T$ at any leaf~$r$. Let~$d_+$ be the maximum indegree of any vertex excluding a possible incoming edge from its parent, i.e., $d_+$ is the maximum number of edges incoming into any vertex~$v$ from children of~$v$; note that~$d_+\leq d-1$. 
For a subtree~$S$ of~$T$ rooted at a vertex~$s$, let~$\ell_{\uparrow}(S)$ be the length of the longest directed path incoming into~$s$ in~$S$. Fix an upward-planar embedding for~$T$ such that, for each vertex~$v$, in clockwise order around~$v$ we have the edge connecting~$v$ with its parent (this does not exist if~$v=r$), then all the edges connecting~$v$ with its children and outgoing from~$v$, and finally all the edges connecting~$v$ with its children and incoming into~$v$.

\begin{figure}[h]
    \centering
    \includegraphics[page=5]{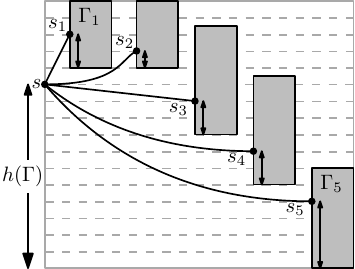}
    \caption{Illustration for the proof of \cref{th:ub-path-degree}. The lower heights are measured by arrows.}
    \label{fig:algorithm-label-app}
\end{figure}

An upward-planar drawing of a rooted tree is \emph{left-anchored} if the root of the tree lies on the left side of the bounding box. We describe a recursive construction, depicted in \cref{fig:algorithm-label-app}, of a left-anchored upward-planar drawing~$\Gamma$ of a subtree~$S$ of~$T$ rooted at a vertex~$s$. When~$S=T$, this gives us the desired drawing of~$T$. Let~$n$ be the number of vertices of~$S$. If~$n=1$, then~$\Gamma$ is trivially constructed. Otherwise, let~$s_1,\dots,s_k$ be the children of~$s$ in~$S$ such that the edges between~$s$ and~$s_1,\dots,s_h$ are outgoing from~$s$ and those between~$s$ and~$s_{h+1},\dots,s_k$ are incoming into~$s$. Recursively construct a drawing~$\Gamma_i$ of each subtree~$S_i$ rooted at a child~$s_i$ of~$s$. We construct~$\Gamma$ by placing the drawings~$\Gamma_1,\dots,\Gamma_k$ side-by-side, in left-to-right order~$\Gamma_1,\dots,\Gamma_k$, so that the following properties are satisfied:
\begin{itemize}
\item the left side of the bounding box of~$\Gamma_1$ is to the right of~$s$;
\item the bottom sides of the bounding boxes of~$\Gamma_1,\dots,\Gamma_h$ are on the same layer, which is one layer above the layer of~$s$, 
\item~$s_{h+1}$ is one layer below the layer of~$s$; and
\item for~$i=h+2,\dots,k$, we have that~$s_i$ is one layer below the layer where the bottom side of the bounding box of~$\Gamma_{i-1}$ lies. 
\end{itemize}

It is easy to see that, since $\Gamma_i$ is a left-anchored upward-planar drawing for $i=1,\dots,k$, then so is $\Gamma$. We define the \emph{lower height}~$h(\Gamma)$ of~$\Gamma$ as the number of layers that intersect~$\Gamma$ and that are not above the layer of~$s$; note that the layer of~$s$ counts towards the value of~$h(\Gamma)$. We now bound~$h(\Gamma)$ in terms of~$d_+$ and~$\ell_{\uparrow}(S)$. If $d_+=0$, then all the drawings~$\Gamma_1,\dots,\Gamma_k$ lie above the layer on which~$s$ lies, hence $h(\Gamma)=1$. Assume next that~$d_+\geq 1$. We prove, by induction on~$\ell_{\uparrow}(S)$, that~$h(\Gamma)\leq \sum_{j=0}^{\ell_{\uparrow}(S)} (d_+)^j$. In the base case, we have~$\ell_{\uparrow}(S)=0$, hence~$s$ has no incoming edge. We then have~$h(\Gamma)=1$, as by construction all the drawings~$\Gamma_1,\dots,\Gamma_k$ lie above the layer on which~$s$ lies. Assume next that~$\ell_{\uparrow}(S)\geq 1$. By construction, we have that~$h(\Gamma)=1+\sum_{i=h+1}^k h(\Gamma_i)$. The key observation here is that~$\ell_{\uparrow}(S_i)\leq \ell_{\uparrow}(S)-1$, given that, for any~$i \in \{h+1,\dots,k\}$, any directed path incoming into~$s_i$ in~$S_i$ becomes a directed path incoming into~$s$ in~$S$ together with the edge~$(s_i,s)$. Thus, by induction, we have~$h(\Gamma)\leq 1+\sum_{i=h+1}^k \sum_{j=0}^{\ell_{\uparrow}(S)-1} (d_+)^j$. Since edges from~$s_{h+1},\dots,s_k$ to~$s$ exist, we have~$k-h\leq d_+$, and hence~$h(\Gamma)\leq 1+d_+\cdot \sum_{j=0}^{\ell_{\uparrow}(S)-1} (d_+)^j=\sum_{j=0}^{\ell_{\uparrow}(S)} (d_+)^j$, as required. Note that~$\sum_{j=0}^{\ell_{\uparrow}(S)} (d_+)^j$ is in~$O((d_+)^{\ell_{\uparrow}(S)})$.

In order to conclude the proof, two things remain to be observed. First,~$\ell_{\uparrow}(S)\leq \ell$,
since~$\ell_{\uparrow}(S)$ is the length of a directed path. Hence, the lower height of any drawing constructed by the algorithm for a subtree of~$T$ is in~$O((d_+)^{\ell})$ and hence in~$O((d-1)^{\ell})$. Second, the edges incoming into~$s$ span only layers that are not above the layer of~$s$, hence they have span smaller than~$h(\Gamma)$; also, an edge outgoing from~$s$ towards a child~$s_i$ has span~$h(\Gamma_i)-1$, as it spans all the layers that contribute to the lower height of~$\Gamma_i$, plus $1$ for the layer of~$s$. Thus, the span of every edge is at most the lower height of a drawing constructed by the algorithm for a subtree of~$T$ and thus it is in~$O((d-1)^{\ell})$.  
\end{proof}

For trees with large degree, the next theorem might give a better bound than \cref{th:ub-path-degree}.

\begin{theorem}
\label{th:ub-longest-path}
Let~$T$ be an~$n$-vertex directed tree whose longest directed path has length~$\ell$. Then~$T$ has span in~$O(\ell \cdot n^{\log_2\varphi})\in O(\ell \cdot n^{0.695})$, where~$\varphi=\frac{1+\sqrt{5}}{2}$ is the golden ratio.
\end{theorem}

\begin{proof}
Assume that $T$ is rooted at any vertex $r$. A \emph{root-to-leaf path} in $T$ is a path $\pi=(r=u_1,u_2,\dots,u_p)$ between $r$ and any leaf $u_p$ of $T$. A root-to-leaf path $\pi$ partitions the subtrees rooted at the children of the vertices in $\pi$ in \emph{up-subtrees} and \emph{down-subtrees}, where a subtree $T_{i,j}$ rooted at a vertex $v_{i,j}$ which is a child of a vertex $u_i$ of $T$ is an up-subtree if the edge $(u_i,v_{i,j})$ is directed from $u_i$ to $v_{i,j}$, and a down-subtree otherwise. A root-to-leaf path $\pi=(r=u_1,u_2,\dots,u_p)$ in $T$ is \emph{greedy} if, for $i=1,\dots,p-1$, the subtree of $T$ rooted at $u_{i+1}$ has at least as many vertices as any other subtree of $T$ rooted at a child of $u_i$. 

In order to design our drawing algorithm, we need the following structural result, which is similar, both in the statement and in the proof, to a result for (non-directed, ordered) trees by Chan \cite{DBLP:journals/algorithmica/Chan02}. The main difference with respect to the result of Chan is that the following lemma applies to trees whose degree is not necessarily bounded by three.

\begin{lemma} \label{le:root-to-leaf-partition}
Let $\pi=(r=u_1,u_2,\dots,u_p)$ be any greedy path in $T$, let $\alpha$ be the number of vertices in any up-subtree of $\pi$ and let $\beta$ be the number of vertices in any down-subtree of $\pi$. Then at least one of the following is true; (i) $\alpha\leq n/2$ and $\beta\leq (n-\alpha)/2$, or (ii) $\beta\leq n/2$ and $\alpha\leq (n-\beta)/2$.
\end{lemma}

\begin{proof}
For any $i=1,\dots,p$, let $n_i$ be the number of vertices in the subtree $T(u_i)$ of $T$ rooted at $u_i$. Let $T_{i,h}$ and $T_{j,k}$ be any up- and down-subtrees of $\pi$ with $\alpha$ and $\beta$ vertices, respectively. Assume that $i\leq j$, as the other case is symmetric\footnote{Since Chan's lemma \cite[Lemma 3.1]{DBLP:journals/algorithmica/Chan02} deals with trees whose degree is at most three, he can assume the inequality between $i$ and $j$ to be strict.}. By the definition of a greedy path, we have $\alpha\leq n_{i+1}$ and $\beta\leq n_{j+1}$. Since the subtrees $T_{i,h}$ and $T(u_{i+1})$ are vertex-disjoint, we get $\alpha + n_{i+1} \leq n$. Also, since the subtrees $T_{i,h}$, $T_{j,k}$, and $T(u_{j+1})$ are vertex-disjoint (this is true even if $j=i$), we get $\alpha + \beta + n_{j+1}\leq n$. Combining $\alpha\leq n_{i+1}$ with $\alpha + n_{i+1} \leq n$ we get $\alpha\leq n/2$. Combining  $\beta\leq n_{j+1}$ with $\alpha + \beta + n_{j+1}\leq n$ we get $\beta\leq (n-\alpha)/2$.
\end{proof}

We now prove that $T$ admits a left-anchored (see the definition on the proof of \cref{th:ub-path-degree}) upward-planar layered drawing $\Gamma$ with at most $(2n^\gamma-1)(\ell+1)$ layers, where $\gamma=\log_2 \varphi=\log_2 (1+\sqrt 5)/2$, which implies that the span of $\Gamma$ is in $O(\ell \cdot n^{\log_2 \varphi})$, as required. The proof is by induction on $n$. In the base case we have $n=1$. Then a left-anchored drawing with $(2n^\gamma-1)(\ell+1)=1$ layer is trivially constructed. In the inductive case, let $\pi=(r=u_1,u_2,\dots,u_p)$ be a greedy root-to-leaf path.

    \begin{figure}[t]
        \centering
        \includegraphics[]{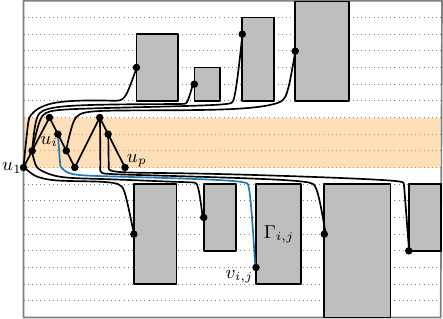}
        \caption{Illustration for the proof of \cref{th:ub-longest-path}. In this example, we have $\ell=3$. The $\ell+1$ layers on which $\pi$ is placed are in the yellow region.}
        \label{fig:ub-longest-path-app}
    \end{figure}

Let $\alpha$ and $\beta$ be the maximum number of vertices of an up- and down-subtree of $\pi$, respectively. The drawing $\Gamma$ of $T$ uses $(2\alpha^\gamma-1)(\ell+1) + (2\beta^\gamma-1)(\ell+1) + (\ell+1)$ layers and is constructed as follows (see \cref{fig:ub-longest-path-app}). 
We recursively construct a left-anchored upward-planar layered drawing $\Gamma_{i,j}$ of each up- and down-subtree $T_{i,j}$. We use the highest $(2\alpha^\gamma-1)(\ell+1)$ layers to place the recursively constructed drawings $\Gamma_{i,j}$ of the up-subtrees of $\pi$ so that that the bottom sides of their bounding boxes lie on the lowest among the $(2\alpha^\gamma-1)(\ell+1)$ layers and so that $\Gamma_{i,j}$ lies to the left of $\Gamma_{i',j'}$, for any $i<i'$ and for any $j$ and $j'$ such that $\Gamma_{i,j}$ and $\Gamma_{i',j'}$ exist. The drawings of the down-subtrees of $\pi$ are placed symmetrically on the lowest $(2\beta^\gamma-1)(\ell+1)$ layers. We place a drawing of $\pi$ on the intermediate $\ell+1$ layers. Note that $\pi=(r=u_1,u_2,\dots,u_p)$ consists of a sequence of directed paths. Each maximal directed path connecting two vertices $u_i$ and $u_j$ with $i<j$ increases in $x$-coordinates from $u_i$ to $u_j$ and increases or decreases in $y$-coordinates from $u_i$ to $u_j$ depending on whether the path is directed from $u_i$ to $u_j$ or vice versa, respectively. Each directed path starts from the lowest and ends at the highest among the $\ell+1$ layers, which ensures that $\pi$ fits on the $\ell+1$ layers, given that each directed path has length at most $\ell$. The drawing of $\pi$ is sufficiently to the left so that $r$ is the leftmost vertex of $\Gamma$. The edges from a vertex $u_i$ of $\pi$ to its children $v_{i,j}$ that are roots of up-subtrees of $\pi$ exit $u_i$ by increasing the $x$-coordinates only slightly, until reaching the highest of the $\ell+1$ layers reserved for $\pi$; from there, they reach the lowest layer among those reserved for the up-subtrees of $\pi$, in a point slightly to the left of the drawing $\Gamma_{i,j}$, so that they can then connect to $v_{i,j}$ without creating crossings. The edges from a vertex $u_i$ of $\pi$ to its children $v_{i,j}$ that are roots of down-subtrees of $\pi$ are drawn similarly. Finally, it remains to prove that $(2\alpha^\gamma-1)(\ell+1) + (2\beta^\gamma-1)(\ell+1) + (\ell+1)\leq (2n^\gamma-1)(\ell+1)$. However, such an inequality is equivalent to $\alpha^{\gamma} +  \beta^{\gamma} \leq n^\gamma$, which is verified in \cite{DBLP:journals/algorithmica/Chan02}, by exploiting the inequalities (i) $\alpha\leq n/2$ and $\beta\leq (n-\alpha)/2$, or (ii) $\beta\leq n/2$ and $\alpha\leq (n-\beta)/2$. These inequalities come from \cite[Lemma 3.1]{DBLP:journals/algorithmica/Chan02} in Chan's result and from \cref{le:root-to-leaf-partition} in our case.
\end{proof}

\section{NP-Hardness}\label{sec:hardness}
We show that the problem $2$-\textsc{span upward planarity} is NP-complete for directed graphs without fixed embedding even if restricted to trees or biconnected graphs with a single source.
We prove NP-hardness by reduction from 3-partition, which is defined as follows.

\textit{\textbf{3-partition:} Given a multi-set $A$ of $3m$ positive integers that sum to $mB$ for some positive integer $B$
    and with $B/4 < a < B/2$ for every $a \in A$, is there a partition of $A$ into $m$ multi-sets $A_1,\dots,A_m$ s.t., %
    for each $i=1,\dots,m$, $A_i$ contains three integers whose sum is~B?}

3-partition is NP-complete even if $B$ is polynomial in $m$ \cite{garey1979computers} (i.e., \emph{strongly} NP-complete).
With both restrictions, we use a similar reduction.

\begin{theorem}
    \label{thm:nph-single-source}
    $2$-\textsc{span upward planarity} is NP-complete even for biconnected single-source graphs.
\end{theorem}

\begin{proof}
    Containment in NP follows from the fact that, given the assignment of vertices to levels, level planarity testing can be done in polynomial time \cite{JungerLM98}
    and it can be verified in linear time that the span is at most 2.
    To show NP-hardness, we reduce from 3-partition.
    Given an instance $A$ of 3-partition with $3m$ integers that sum to $mB$, we define a directed graph~$G$.
    Some parts of $G$ model on the one hand pockets and frames,
    and on the other hand gadgets of a specific height,
    each representing a number of~$A$.
    
    \begin{figure}[t]
        \begin{subfigure}[b]{.25\textwidth}
            \centering
            \includegraphics[page=2]{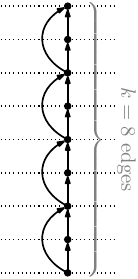}
            \subcaption{}
            \label{fig:fixed-length-path}
        \end{subfigure}
        \hfill
        \begin{subfigure}[b]{.7\textwidth}
            \centering
            \includegraphics[page=3]{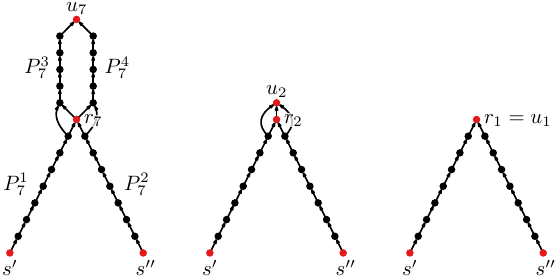}
            \subcaption{}
            \label{fig:number-gadget}
        \end{subfigure}
        \caption{(a) Fixed-length path. (b) Number gadgets for the numbers 7, 2, and 1.}
    \end{figure}
    
    A key building block is a \emph{fixed-length path (of length $k$)}, 
    that is, a pair of directed parallel paths that share only their endpoints:
    the one path contains $k$ edges for some even $k \ge 2$ and is called \emph{long path},
    while the other path contains $k/2$ edges and is called \emph{short path}.
    Observe that in any drawing of span 2, the fixed-length path can only be realized with the
    long path having only edges of span~1 and the short path having only edges of span~2.
    Hence, such a path bridges $k$ levels (and has vertices on $k+1$ levels).
    For an illustration, see \cref{fig:fixed-length-path}.
    
    \begin{figure}[t]
        \centering
        \includegraphics[page=1]{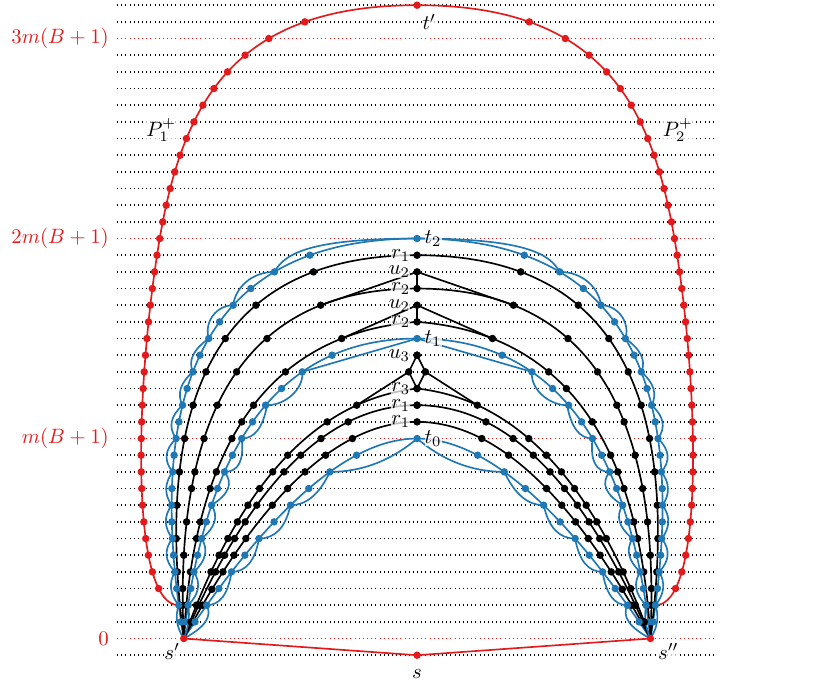}
        \caption{The biconnected single-source graph for the instance $M=\{1,1,1,2,2,3\}$ of 3-partition (where $m = 2$ and $B = 5$) and an upward-planar layered drawing that corresponds to the partition $\{1,1,3\},\{2,2,1\}$.
            Frame and pockets are red and blue, while the number gadgets are black.
        }
        \label{fig:nphard-single-source}
    \end{figure}
    
    We are now ready to define the graph~$G$.
    We start by describing the frame,
    which is shown in blue and red in \cref{fig:nphard-single-source}.
    There is a single source $s$ that has two outgoing edges leading to vertices $s'$ and $s''$.
    For $i = 0, 1, 2, \dots, m$, we add two fixed-length path.
    One starts in~$s'$, the other starts in~$s''$.
    Both fixed-length paths terminate in the same vertex $t_i$ and get length $(m + i)(B+1)$.
    To the third vertex on the long path from $s'$ to $t_m$,
    we add a directed path~$P^+_1$ of length $3m (B+1)$ terminating in a vertex~$t'$.
    Symmetrically, we add a directed path~$P^+_2$ of length $3m (B+1)$ terminating in the vertex~$t'$
    to the third vertex on the long path from $s''$ to $t_m$.
    This finishes the description of the frame.
    
    Now for every number $a \in A$,
    we use a \emph{number gadget} $U_a$;
    see \cref{fig:number-gadget}.
    It consists of
    \begin{itemize}
        \item a directed path $P_a^1$ with $m(B+1)$ edges starting in $s'$ and terminating in a vertex~$r_a$,
        \item a directed path $P_a^2$ with $m(B+1)$ edges starting in $s''$ and terminating also in~$r_a$,
        \item two parallel directed paths $P_a^3$ and $P_a^4$ with $a$ vertices (and $a-1$ edges)
        that both start in $r_a$ and both terminate in a vertex $u_a$, and
        \item two edges, one connecting the predecessor and successor of $r_a$ along $P_a^1$ and $P_a^3$, respectively, 
        one connecting the predecessor and successor of $r_a$ along $P_a^2$ and $P_a^4$, respectively.
    \end{itemize}
    
    Clearly, $s$ is the only source of~$G$.
    Moreover, note that any two edges lie
    on a cycle and, hence, $G$ is biconnected.
    Clearly, $G$ has polynomial size in~$O(mB)$
    and can be described in the same time bound.
    This is sufficient for the strongly NP-complete problem 3-partition.
    
    We will next argue that $G$
    admits an upward-planar drawing with span~2 if and only if
    the given instance $A$ of 3-partition is a yes-instance.
    
    \subparagraph*{Yes-instance of 3-partition $\Rightarrow$ span-2 drawing.}
    Let $A_1, \dots, A_m$ be a solution to the 3-partition instance~$A$.
    We show how to construct an upward planar drawing of~$G$ with span~2 using this solution.
    The main idea is to group the number gadgets according to this solution inside the frame
    to obtain a drawing like in \cref{fig:nphard-single-source}.
    We now describe this drawing in more detail.
    Place $s$ on level $-1$ and place $s'$ and $s''$ on level 0.
    Draw the frame as illustrated in \cref{fig:nphard-single-source}, that is,
    for $i \in \{0, 1, \dots, m\}$,
    the endpoint $t_i$ of the fix lengths path starting in $s'$ and~$s''$ lies on level~$(m+i)(B+1)$.
    For $i \in [m]$, draw the number gadgets of $A_i$
    in the pockets between the fixed-length paths terminating
    in $t_{i-1}$ and $t_i$.
    More precisely, for $a_1, a_2, a_3 \in A_i$,
    place $r_{a_1}$ on level $(m+i-1)(B+1)+1$,
    $u_{a_1}$ on level $(m+i-1)(B+1)+a_1$,
    $r_{a_2}$ on level $(m+i-1)(B+1)+a_1+1$,
    $u_{a_2}$ on level $(m+i-1)(B+1)+a_1+a_2$,
    $r_{a_3}$ on level $(m+i-1)(B+1)+a_1+a_2+1$,
    $u_{a_3}$ on level $(m+i-1)(B+1)+B = (m+i)(B+1)-1$.
    Observe that this allows $P_{a_\star}^3$ and $P_{a_\star}^4$ ($a_\star \in a_1, a_2, a_3)$ to be drawn.
    Furthermore, note that~$P_{a_\star}^1$ and~$P_{a_\star}^2$
    can be drawn if we combine edges of span~1 and~2.
    Recall that these paths have $m (B+1)$ edges.
    Let $U_{a_b}$ and $U_{a_t}$ be
    the bottommost and topmost number gadget, respectively.
    Then, $r_{a_b}$ is placed on level $m(B+1)+1$,
    while $r_{a_t}$ is placed strictly below level~$2m(B+1)$.
    This also allows, for each $a \in A$, that
    the predecessors/successors of each~$r_a$
    is drawn 1 level below/above~$r_a$ to draw the extra edges with span~2.
    Finally observe that $P^+_1$ and $P^+_2$
    can be drawn in the outer face of the remaining drawing
    with~$y(t') = 3m(B+1) + 2$.

    \subparagraph*{Span-2 drawing $\Rightarrow$ yes-instance of 3-partition.}
    Let an upward-planar drawing of~$G$ having span~2 be given.
    There is only one level assignment
    possible for the fixed-length paths, which bound the pockets.
    Thus, the endpoints $s'$ and $s''$ of the fixed-length paths
    lie on the same level, which is w.l.o.g.\ level~0,
    and $y(t_i) = (m+i)(B+1)$ for $i = \{0, 1, \dots, m\}$.
    Then, $s$ lies on level $-1$ or $-2$,
    and $y(t') \in \{3m(B+1)+2, 6m(B+1)+2\}$.
    In any case, since $y(t') > y(t_m)$,
    $P^+_1$ and $P^+_2$ are drawn outside the pockets
    above the fixed-length paths as illustrated in \cref{fig:nphard-single-source}.
    
    It remains to analyze the drawing of the number gadgets.
    First observe, since the number gadgets and
    the fixed-length paths forming the pockets are connected
    and terminate in the same vertices $s'$ and~$s''$,
    each number gadget is drawn either completely above
    or completely below two fixed-length paths terminating in a common vertex~$t_i$.
    Now observe that for every $a \in A$, the number gadget~$U_a$ is drawn below~$t_m$
    because otherwise $U_a$ would intersect $P^+_1$ and~$P^+_2$
    as $t'$ lies above $u_a$.
    Moreover, $r_a$ lies above~$t_0$ since the $s'$--$r_a$ path
    has $m(B+1)$ edges and $t_0$ lies on level $m(B+1)$.
    Therefore, every number gadget can be associated
    with precisely one pocket,
    where the $i$-th pocket (for $i \in [m]$) is the area between
    the fixed-length paths terminating in $t_{i-1}$ and~$t_i$.
    Let the multisets $A_1, A_2, \dots, A_m$
    be the numbers whose number gadgets are
    drawn in the 1st, 2nd, $\dots$, $m$-th pocket, respectively.
    Note that, for every $a \in A$, $U_a$
    requires $a$ levels for the paths $P_a^3$ and~$P_a^4$
    between~$r_a$ and~$u_a$
    that are distinct from those used by the other
    number gadgets and from every $t_i$ for $i \in \{0, 1, \dots, m\}$.
    For each $i \in [m]$,
    there are $B$ levels strictly between $t_{i-1}$ and $t_i$.
    Over all pockets, these are $mB$ levels,
    which precisely matches $\sum_{a\in A} a$.
    Hence, the sum of numbers in each pocket is $B$
    and $A_1, \dots, A_m$ is a 3-partition.
\end{proof}

We use a similar NP-hardness construction for directed trees,
but we cannot exploit the rigid structure of fixed-length paths;
see \cref{fig:nphard-trees}.

\begin{figure}[t]
    \centering
    \includegraphics[page=2,width=\textwidth]{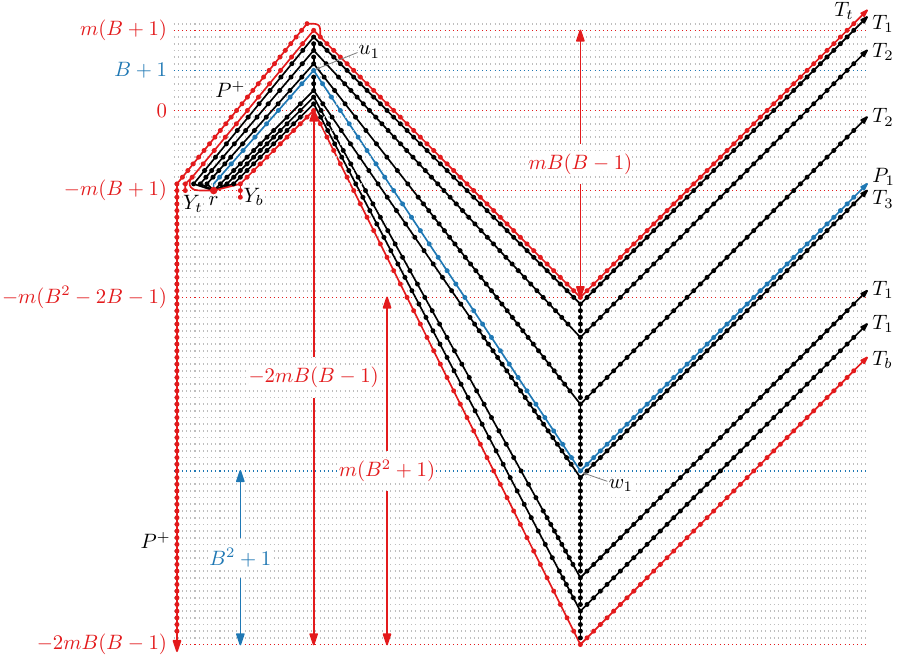}
    \caption{Directed tree from \cref{thm:nph-trees}
        for the same 3-partition instance as in \cref{fig:nphard-single-source}.
        Again, frame and pockets are red and blue, while the number gadgets are black.
    }
    \label{fig:nphard-trees}
\end{figure}

\begin{theorem}
    \label{thm:nph-trees}
    $2$-\textsc{span upward planarity} is NP-complete for directed trees.%
\end{theorem}

\begin{proof}
    As for \cref{thm:nph-single-source}, the problem is in NP.
    It is NP-hard by reduction from 3-partition.
    
    \subparagraph*{Construction.}
    Let $A$ be an instance of 3-partition with $B= \frac1m\sum_{a \in A} a$ and let $m = |A| / 3$.
    We may assume that $B$ is odd.
    Otherwise replace every integer $a \in A$ by $a+1$.
    We construct a directed tree $T$ from $A$; see \cref{fig:nphard-trees}.
    Tree $T$ contains a \emph{central vertex} $r$ of degree $3m+m+1$.
    All other vertices have degree at most three.
    Vertex $r$ splits $T$ into subtrees;
    $3m$ of which are called \emph{integer trees} and
    each integer tree~$T_a$ corresponds to an integer $a \in A$ (black vertices in \cref{fig:nphard-trees}),
    one is called \emph{bottom tree}~$T_b$, one \emph{top tree}~$T_t$ (red vertices in \cref{fig:nphard-trees}),
    and the other $m-1$ trees are \emph{separating paths} $P_i$, $i=1,\dots,m-1$ (blue vertices in \cref{fig:nphard-trees}).
    Starting from $r$, the separating path $P_i$ has $(m+i)(B+1)/2$ ascending edges (ending at a sink vertex~$u_i$), followed by $B(B-1)(2m-i)/2$ descending edges (ending at a source vertex~$w_i$), and then $2mB^2$ ascending edges.
    The bottom tree~$T_b$, the top tree~$T_t$, and every integer tree~$T_a$ contain a \emph{central path} $\langle r=v_x^0,\dots,v_x^{m(3B^2+1)} \rangle$ each,
    where $x \in \{b,t\} \cup A$, with $m(B+1)$ ascending edges, followed by $mB(B-1)$ descending edges, and then $2mB^2$ ascending edges.
    The bottom tree and the top tree have in addition directed paths $Y_b$ and~$Y_t$ with two edges, called \emph{outer stub}
    (see on the left side of \cref{fig:nphard-trees}).
    The outer stubs $Y_b$ and $Y_t$ end at the second vertex~$v^1_b$ and~$v^1_t$ of the central paths of $T_b$ and $T_t$, respectively.
    In addition, $T_t$ contains a path $P^+$
    that starts at $v_t^{m(B+1)+1}$ (i.e., one vertex after the sink vertex of the central path of~$T_t$)
    has one edge directed upwards followed by
    $3mB^2$ edges directed downwards.
    For each $a \in A$,
    the integer tree $T_a$ has in addition two directed paths called $Z_a$ and $Z'_a$;
    $Z_a$ has size\footnote{Here, the size of a path is the number of its vertices.}~$a$
    and its bottommost vertex is the sink~$v_a^{m(B+1)}$,
    and $Z'_a$ has size $aB$ and its topmost vertex is the source $v_a^{m(B^2+1)}$. %
    
    \subparagraph*{Correctness.}
    We show that the tree $T$ has span 2 if and only if $A$ is a yes instance of 3-partition.
    If there is a solution to~$A$, the high-level idea is to make all subtrees of~$r$ form a nested up--down--up zigzag pattern going to the right
    as illustrated in \cref{fig:nphard-trees}.
    The subtrees rooted at~$r$ are permuted in such a way that~$T_b$ and~$T_t$
    are the outermost trees and the separating paths $P_1, \dots, P_{m-1}$ form pockets of height $B$ at the ``high point''
    and pockets of height $B^2$ at the ``low point''.
    Into each pocket, we arrange three trees that represent three numbers from~$A$ being grouped together in the solution of~$A$.
    Note that, due to the sizes of the additional paths $Z_a$ and $Z'_a$,
    these three trees need precisely $B$ levels at their ``high point'' and $B^2$ levels at their ``low point''.
    In the other direction, we are given a drawing of~$T$ with span~2.
    We will analyze how such a drawing can look like due to the structure of~$T$.
    We will arrive at the conclusion that the only way to draw~$T$ with span~2 is
    in a similar manner as in \cref{fig:nphard-trees} (potentially mirrored).
    The separating paths $P_1, \dots, P_{m-1}$ do again form pockets of height $B$ and $B^2$.
    All levels of these pockets need to be occupied by trees corresponding to numbers in~$A$.
    Since the number of levels in a pocket is~$B$ (or~$B^2$),
    this corresponds to groups of numbers in~$A$ summing up to~$B$.
    In other words, we can read off a solution for~$A$ from a drawing of~$T$ having span~2.
    
    \subparagraph*{Yes-instance of 3-partition $\Rightarrow$ span-2 drawing.}
    Let $A_1, \dots, A_m$ be a solution to the 3-partition instance~$A$.
    We show how to construct an upward planar drawing of~$T$ with span~2 using this solution.
    In the drawing, we order $r$'s subtrees at~$r$ as follows from right to left (see \cref{fig:nphard-trees}).
    The rightmost subtree is~$T_b$.
    Then, for each $i \in [m]$, the trees $T_{a_i^1}$, $T_{a_i^2}$, $T_{a_i^3}$, $P_i$ follow in this order from right to left where $a_i^1$, $a_i^2$, and $a_i^3$ are the three elements of~$A_i$.
    The leftmost subtree is~$T_t$.
    Since $T_b$ and $T_t$ are the outermost trees, we can draw $Y_b$, $Y_t$, and $P^+$ without crossings by attaching them to the outside;
    in particular, the edge of $P^+$ incident to~$v_t^{m(B+1)+1}$
    has span 2 such that the next vertex is one level above $v_t^{m(B+1)}$
    and the rest of $P^+$ leads downwards on the left side of the drawing.
    
    Let $r$ be placed on level $-m(B+1)$.
    Then, place the sink $v_b^{m(B+1)}$ on level~0.
    Note that this is possible by assigning every edge of the central path of~$T_b$ from $r  = v_b^0$ to $v_b^{m(B+1)}$ span~1.
    Similarly, we place the sink $v_t^{m(B+1)}$ on level~$m(B+1)$.
    Note that this is possible by assigning every edge of the central path of~$T_t$ from $r  = v_t^0$ to $v_t^{m(B+1)}$ span~2.
    In between, we place, for each $i \in [m]$, the sink~$v_{a_i^1}^{m(B+1)}$ on level~$i (B+1) - B$,
    the sink~$v_{a_i^2}^{m(B+1)}$ on level~$i (B+1) - B + a_i^1$,
    and the sink~$v_{a_i^3}^{m(B+1)}$ on level~$i (B+1) - B + a_i^1 + a_i^2$.
    Note that this is possible by combining edges of spans~1 and~2 on the central paths of~$T_{a_i^1}$, $T_{a_i^2}$, and~$T_{a_i^3}$
    from $r$ to to the sinks $v_{a_i^1}^{m(B+1)}$, $v_{a_i^2}^{m(B+1)}$, and~$v_{a_i^3}^{m(B+1)}$.
    Moreover, we draw, for each $a \in A$, every edge of~$Z_a$ with span~1;
    hence, the vertices of $Z_a$ do not share a level with any sink $v_{a'}^{m(B+1)}$ of the central path of a different tree~$T_a'$ with $a' \in A$
    (see on the top left in \cref{fig:nphard-trees}).
    It remains to assign (the first portions of) the separating paths $P_1, \dots, P_{m-1}$.
    For each $i \in [m-1]$, we place the sink $u_i$ at level $i (B+1)$,
    which is possible if every edge between $r$ and $u_i$ has span~2.
    Furthermore, observe that $u_i$ does not share a level with any sink $v_{a}^{m(B+1)}$
    of the central path of a tree~$T_a$ with $a \in A$ or with a vertex of~$Z_a$.
    
    All sinks of the central paths $v_b^{m(B+1)}$, $v_{a_1^1}^{m(B+1)}$, $v_{a_1^2}^{m(B+1)}$, $v_{a_1^3}^{m(B+1)}$, $v_{a_2^1}^{m(B+1)}$, \dots, $v_{a_m^3}^{m(B+1)}$, $v_t^{m(B+1)}$ and all sinks $u_1, \dots, u_{m-1}$
    have as their left child the next vertex on the path towards $r$ and as their right child the next vertex on the path towards
    the sources (w.r.t.\ the central paths) $v_b^{m(B^2+1)}$, $v_{a_1^1}^{m(B^2+1)}$, $v_{a_1^2}^{m(B^2+1)}$, $v_{a_1^3}^{m(B^2+1)}$, $v_{a_2^1}^{m(B^2+1)}$, \dots, $v_{a_m^3}^{m(B^2+1)}$, $v_t^{m(B^2+1)}$
    and the sources $w_1, \dots, w_{m-1}$.
    We assign the source $v_b^{m(B^2+1)}$ to level $-2mB(B-1)$.
    Note that this is possible if each edge between the sink~$v_b^{m(B+1)}$ and the source~$v_b^{m(B^2+1)}$ has span~2 since $v_b^{m(B+1)}$ is on level 0:
    \begin{equation*}
        y(v_b^{m(B^2+1)}) = 0 - 2 \cdot (m(B^2+1) - m(B+1)) = - 2 (m B^2 - mB) = -2 mB (B-1).
    \end{equation*}
    Similarly, we assign the source $v_t^{m(B^2+1)}$ to level $-m (B^2 - 2B - 1)$.
    Note that this is possible if each edge between the sink~$v_t^{m(B+1)}$ and the source~$v_t^{m(B^2+1)}$ has span~1 since $v_t^{m(B+1)}$ is on level $m(B+1)$:
    \begin{equation*}
        y(v_t^{m(B^2+1)}) = m(B+1) - (m(B^2+1) - m(B+1))
        = -m (B^2 - 2B - 1).
    \end{equation*}
    For each $i \in [m-1]$, we assign $w_i$ to level~$-2mB(B-1) + i (B^2+1)$.
    This is possible by assigning each edge between sink~$u_i$ and source~$w_i$ span~2.
    Recall that there are $B(B-1)(2m-i)/2$ edges in between and $u_i$ lies on level $i(B+1)$. Hence,
    \begin{equation*}
        y(w_i) = i(B+1) - 2 \cdot \frac{B(B-1)(2m-i)}{2}
        = -2mB (B - 1) + i (B^2+1).
    \end{equation*}
    For each $i \in [m]$, let $a_i^1$, $a_i^2$, and $a_i^3$ be the three numbers of~$A_i$.
    We assign the source~$v_{a_i^1}^{m(B^2+1)}$ to level $-2mB(B-1) + (i-1) (B^2+1) + a_i^1 B$
    and each of the $a_i^1B-1$ other vertices of $Z'_{a_i^1}$ to the next levels below.
    Similarly, we assign the sources $v_{a_i^2}^{m(B^2+1)}$ and $v_{a_i^3}^{m(B^2+1)}$
    to the levels $-2mB(B-1) + (i-1) (B^2+1) + (a_i^1 + a_i^2) B$
    and $-2mB(B-1) + (i-1) (B^2+1) + B^2$, respectively.
    The $a_i^2B-1$ and $a_i^3B-1$ other vertices of $Z'_{a_i^2}$ and $Z'_{a_i^3}$ follow on the next levels below $v_{a_i^2}^{m(B^2+1)}$ and $v_{a_i^3}^{m(B^2+1)}$, respectively.
    Observe that the levels of the latter vertices
    and $w_1, \dots, w_{m-1}$ are distinct.
    However, it remains to analyze if the downward parts of the central paths of $T_{a_i^1}$, $T_{a_i^2}$, and $T_{a_i^3}$
    can reach the desired levels if all edges have span~1 or~2.
    Recall that the sink $v_{a_i^1}^{m(B+1)}$ lies on level~$i (B+1) - B$ and the downward path to the source~$v_{a_i^1}^{m(B^2+1)}$ has $mB(B-1)$ edges.
    Indeed, for all edges having span~1, it holds that
    \begin{align*}
        y(v_{a_i^1}^{m(B^2+1)}) = -2mB(B-1) + (i-1) (B^2+1) + a_i^1 B &\le i (B+1) - B - mB(B-1) \\
        \Leftrightarrow (i-1) (B^2+1) + a_i^1 &\le mB(B-1) + iB \\
        \Leftrightarrow iB^2 + i + mB + a_i^1 &\le mB^2 + iB + B + 1 \\
        \Leftrightarrow (m-i-1) B + i + a_i^1 &\le mB^2 + 1,
    \end{align*}
    and, for all edges having span~2, it holds that
    \begin{align*}
        y(v_{a_i^1}^{m(B^2+1)}) = -2mB(B-1) + (i-1) (B^2+1) + a_i^1 B &\ge i (B+1) - B - 2mB(B-1) \\
        \Leftrightarrow (i-1) (B^2+1) + a_i^1 B &\ge iB.
    \end{align*}
    The former is true because $(m-i-1) B + i \le mB$ and $a_i^1 \le B$.
    The latter is true because if $i = 1$, then $a_i^1 B \ge B$, and otherwise $(i-1) (B^2+1) \ge iB$.
    Similar arguments apply to $v_{a_i^2}^{m(B+1)}$ and $v_{a_i^3}^{m(B+1)}$:
    \begin{align*}
        y(v_{a_i^2}^{m(B^2+1)}) &= -2mB(B-1) + (i-1) (B^2+1) + (a_i^1 + a_i^2) B \le -mB (B-1) + iB + a_i^1; \\
        y(v_{a_i^3}^{m(B^2+1)}) &= -2mB(B-1) + (i-1) (B^2+1) + B^2 \le -mB (B-1) + iB + a_i^1 + a_i^2.
    \end{align*}
    This is still true because also $a_i^1 + a_i^2 \le B$ and $B \le B$ and, clearly,
    \begin{align*}
        y(v_{a_i^2}^{m(B^2+1)}) &= -2mB(B-1) + (i-1) (B^2+1) + (a_i^1 + a_i^2) B \ge -2mB (B-1) + iB + a_i^1; \\
        y(v_{a_i^3}^{m(B^2+1)}) &= -2mB(B-1) + (i-1) (B^2+1) + B^2 \ge -2mB (B-1) + iB + a_i^1 + a_i^2.
    \end{align*}
    
    Furthermore, at the sources $v_t^{m(B^2+1)}$, $v_{a_1^1}^{m(B^2+1)}$, $v_{a_1^2}^{m(B^2+1)}$, \dots, $v_{a_m^3}^{m(B^2+1)}$, $v_t^{m(B^2+1)}$
    and $w_1, \dots, w_{m+1}$, the right outgoing edge leads away from~$r$.
    Hence, all trees can extend to the right side (using span 1 for every edge) arbitrarily far as shown in \cref{fig:nphard-trees}.
    Since only outgoing edges follow, this results in an upward-planar drawing of~$T$ that has span~2.
    
    \subparagraph*{Span-2 drawing $\Rightarrow$ yes-instance of 3-partition.}
    Let an upward-planar drawing of~$T$ having span~2 be given.
    Consider the radial order of (outgoing) subtrees of~$r$.
    We first show that the two outermost subtrees are~$T_b$ and~$T_t$.
    Suppose that one of~$T_b$ and~$T_t$, say~$T_t$ is an inner subtree.
    As $v_t^1$ lies one or two levels above~$r$,
    the bottommost vertex of~$Y_t$ lies on the same level as~$r$ or below~$r$.
    This means there is a crossing with the edges of the neighboring subtrees of~$T_t$
    since they are all rooted at $r$; a contradiction.
    
    So, w.l.o.g., let $T_t$ be the leftmost and $T_b$ be the rightmost subtree of~$r$
    (otherwise mirror the whole drawing).
    Consider $T_t$.
    Since $v_t^{m(B+1)}$ is a sink of $T_t$,
    its neighbor $v_t^{m(B+1)+1}$ must be on a lower level.
    Note that $v_t^{m(B+1) + 1}$ has two more neighbors:
    one more on the central path
    (when following it, it leads downwards and then far up), 
    and a neighbor on $P^+$
    (when following $P^+$, it leads one edge up and then far down).
    Since $P^+$ cannot cross the rest of the central path,
    one of~$P^+$ and the central path of~$T_t$ passes $r$ on the left, while the other one
    passes~$r$ on the right~--
    in particular, this means that the sink of~$P^+$
    must be on the level above the sink~$v_t^{m(B+1)}$;
    see \cref{fig:nphard-trees}.
    Moreover, the (central) paths of all other subtrees of~$r$
    must also not cross $P^+$.
    Hence, they pass $r$ on the same side as the central path of~$T_t$.
    More precisely, the sink~$v_t^{m(B+1)}$ has as its left child~$v_t^{m(B+1)-1}$ and has as its right child $v_t^{m(B+1)+1}$, and, thus,
    $P^+$ passes $r$ on the left while
    all other (central) paths of subtrees of~$r$
    pass~$r$ on the right side.
    Suppose for a contradiction that $P^+$ passes $r$ on the right and all other
    (central) paths of subtrees of~$r$ pass~$r$ on the left side.
    Since $T_b$ is the rightmost subtree of~$r$,
    the sink~$v_b^{m(B+1)}$ is on a higher level than~$v_t^{m(B+1)}$.
    This would imply a crossing between~$T_b$ and~$P^+$.
    
    The paths $P_1, \dots, P_{m-1}$
    and the central paths of~$T_b$, $T_t$, and every $T_a$ for~$a \in A$
    have nested sinks and sources as in \cref{fig:nphard-trees}.
    Suppose for a contradiction that the middle part of some central path~$X$
    (between the inner sink and the inner source)
    passes one of the sinks of another central path, say~$v_b^{m(B+1)}$,
    on the left side instead of the right side
    (which means their sinks would not nest).
    When traversing $X$ starting at $r$, then even if
    all edges between~$r$ and the first sink have span~2 and
    all edges between this sink and the next source have span~1,
    there still would be a crossing between $X$
    and an edge of~$T_b$ between~$r$ and the sink $v_b^{m(B+1)}$
    due to the lengths of the parts going first upwards and then downwards.
    By the same argument, all sources nest
    because the lengths of the final upward parts
    are sufficiently much longer than the middle downward parts.
    
    Now note that, because of this nested structure,
    all sinks of the paths $P_1, \dots, P_{m-1}$
    and the central paths of~$T_b$,~$T_t$, and every~$T_a$ for~$a \in A$
    as well as all vertices of~$Z_a$ for~$a \in A$
    need to lie on distinct levels.
    The same holds for all sources of those paths
    as well as all vertices of~$Z'_a$ for~$a \in A$.
    Hence, the former need $(m-1) + 2 + mB = m (B + 1) + 1$ distinct levels and
    the latter need $(m-1) + 2 + mB^2 = m (B^2 + 1) + 1$ distinct levels.
    Consider the nested sinks and recall that
    $v_b^{m(B+1)}$ is the bottommost and $v_t^{m(B+1)}$ is the topmost of them.
    W.l.o.g., let $v_b^{m(B+1)}$ lie on level~0.
    Then, $v_t^{m(B+1)}$ lies on level $m (B + 1)$ or higher.
    Since both the $r$--$v_b^{m(B+1)}$ path
    and the $r$--$v_t^{m(B+1)}$ path have $m(B+1)$ edges,
    this is only possible if $r$ lies on level $-m(B+1)$,
    all edges on the $r$--$v_b^{m(B+1)}$ path have span~1 and
    all edges on the $r$--$v_t^{m(B+1)}$ path have span~2.
    Then, $v_t^{m(B+1)}$ lies precisely on level $m(B+1)$.
    Based on this we can assume, w.l.o.g.,
    that the source $v_b^{m(B^2+1)}$ lies on the lowest reachable level and
    that the source $v_t^{m(B^2+1)}$ lies on the highest reachable level.
    Both, the $v_b^{m(B+1)}$--$v_b^{m(B^2+1)}$ path and
    the $v_t^{m(B+1)}$--$v_t^{m(B^2+1)}$ path have $m(B^2+1)-m(B+1) = mB(B-1)$ edges.
    Therefore, given that all edges of the
    $v_b^{m(B+1)}$--$v_b^{m(B^2+1)}$ path have span~2,
    $v_b^{m(B^2+1)}$ lies on level~$-2m(B-1)$.
    Furthermore, given that all edges of the
    $v_t^{m(B+1)}$--$v_t^{m(B^2+1)}$ path have span~1,
    $v_t^{m(B^2+1)}$ lies on level $m(B+1) - mB(B-1) = -m (B^2 - 2B - 1)$.
    
    We conclude that at the ``high point'',
    the sink~$v_t^{m(B+1)}$ of~$T_t$ lies has high as possible and
    the sink~$v_b^{m(B+1)}$ of~$T_b$ lies has low as possible
    as well as at the ``low point'',
    the source~$v_t^{m(B^2+1)}$ of~$T_t$ lies has high as possible and
    the source~$v_b^{m(B^2+1)}$ of~$T_b$ lies has low as possible,
    and, in both cases, all levels in between are occupied by
    distinct vertices of $Z_a$ and the sinks of all other central paths
    as well as distinct vertices of $Z'_a$ and the source of all other central paths.
    
    We next show that, for every $i \in [m-1]$, all edges of the separating path~$P_i$ between the source $r$ and the sink~$u_i$
    as well as all edges between the source~$w-i$ and the sink~$u_i$ have span~2.
    First suppose for a contradiction that some edge on~$P_1$ between $r$ and~$u_1$ has span~1.
    Then, there are not $B$ levels between the sinks~$u_1$ and~$v_b^{m(B+1)}$
    but only $C < B$ levels.
    Let $A' \subseteq A$ be a multiset of at most three integers
    such that the $Z_{a'}$ for all $a' \in A'$
    occupy the $C$ levels between $u_1$ and~$v_b^{m(B+1)}$.
    Clearly, $\sum_{a \in A'} = C$.
    Hence, between the sources $w_1$ and~$v_b^{m(B^2+1)}$,
    there are precisely vertices from $Z'_{a'}$ with $a' \in A'$.
    Observe that these $Z'_{a'}$ need $CB$ distinct levels.
    Therefore, since the source~$v_b^{m(B^2+1)}$ lies on level~$-2mB(B-1)$,
    $w_1$ lies on level
    \begin{align*}
        y(w_1) &= -2mB(B-1) + CB + 1 \overset{(C \le B-1)}{\le} -2mB(B-1) + (B-1)B + 1 \\
        &= - (2m-1) B (B-1) + 1.
    \end{align*}
    On the other hand, every edge on~$P_1$ between
    $u_1$ and~$w_1$ has span at most 2.
    Hence,
    \begin{equation*}
        y(w_1) \ge C + 1 - 2 \cdot (B(B-1)(2m-1)/2)
        = - (2m-1) B (B-1) + 1 + C.
    \end{equation*}
    This is a contradiction since $C > 0$.
    Hence, all edges on~$P_1$ between $r$ and~$u_1$
    have span~2 and~$u_1$ lies on level $y(u_1) = B + 1$.
    The $B$ levels between $r$ and~$u_1$ are occupied by
    three $Z_a$ such that $a \in A^\star \subseteq A$
    and $\sum_{a \in A^\star} a = B$.
    Then, $w_1$ lies on level
    \begin{equation*}
        y(w_1) \ge B + 1 - 2 \cdot (B(B-1)(2m-1)/2)
        = - 2mB(B-1) + B^2 + 1.
    \end{equation*}
    This leaves at least $B^2$ levels strictly between the
    sources~$v_b^{m(B^2+1)}$ and~$w_1$.
    Since the $Z'_a$ with $a \in A^\star$ need precisely $B^2$ levels,
    we know that $w_1$ lies precisely on level $- 2mB(B-1) + B^2 + 1$.
    
    Now that we have determined the levels occupied by $P_1$,
    we apply the same argument inductively for $P_2, P_3, \dots$.
    Hence, the separating paths $P_1, \dots, P_{m-1}$
    form bags of height~$B$ at the ``high point''
    and bags of height~$B^2$ at the ``low point''.
    These bags divide the integer trees in groups of three,
    each of which corresponds to integers summing to~$B$.
    Therefore, given a drawing of span~2,
    we can always read off a solution
    of the original 3-partition instance.
\end{proof}

\section{Algorithms for Graphs with a Bounded Number of Sources}\label{sec:algorithms}

    \begin{figure}[tb]
    \centering
    \begin{subfigure}{0.33\textwidth}
	\centering
	\includegraphics[page=4]{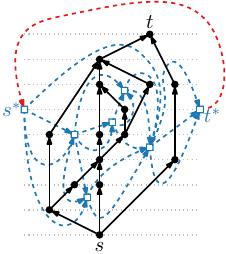}
    
    \subcaption{\label{fig:flowNetwork}}
    \end{subfigure}
\hfil
    \begin{subfigure}{0.3\textwidth}
	\centering
	\includegraphics[page=6]{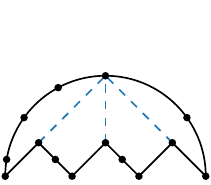}
    
	\subcaption{\label{fig:singleSource-inner}}
	\end{subfigure} 
\begin{subfigure}{0.3\textwidth}
	\centering
	\includegraphics[page=5]{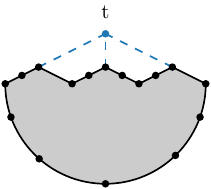}
    
    \subcaption{\label{fig:singleSource-outer}}
	\end{subfigure}
        \caption{(a) The flow network used to decide whether a plane st-graph has span at most $k$. (b)~The augmentation in the proof of \cref{th:single-source-upward} for an internal face and (c) \mbox{for the outer face.}}
        \label{fig:algorithmst-label}
    \end{figure}

In this section, we provide algorithms to solve 
{\sc $k$-span leveled planarity} for several graph classes.
For plane st-graphs, we use a network-flow with upper and lower capacity bounds on the edges, inspired by
the compaction of orthogonal drawings with rectangular shapes minimizing the total edge length \cite{gdbook}. The edges of this network are dual to those of the graph and the flow traversing them corresponds to the span of the dual edge; see \cref{fig:flowNetwork}.

\begin{lemma}
    \label{th:st-plane}
    For an $n$-vertex plane st-graph $G$ and any function \mbox{$\sigma: E \rightarrow \mathbb{N}$}, we can decide in $O(n)$ time whether $G$ has an embedding-preserving upward-planar layered drawing such that the span of each edge $e \in E$ is at most $\sigma(e)$.
\end{lemma}

\begin{proof}
    We show that the problem of the statement reduces in linear time to testing the existence of a feasible circulation in a network with lower and upper capacity bounds. 

    The boundary of every face $f$ of a plane st-graph consists of two directed paths connecting the \emph{face source} $s_f$ of $f$ and the \emph{face sink} $t_f$ of $f$, i.e., the unique vertices on the boundary of $f$ with no incoming edges and no outgoing edges on that boundary, respectively.
    The \emph{left path} (resp. \emph{right path}) of $f$ is the directed path, from $s_f$ to $t_f$, encountered when traversing the boundary of $f$ clockwise (resp.\ counter-clockwise) starting from the face source.
    
    Let $s$ be the source and let $t$ be the sink of $G$. We construct a directed flow network $\cal N$ on the weak dual $G^*$ of $G$ as shown in \cref{fig:flowNetwork}. The vertices correspond to the internal faces of $G$ together with two special vertices $s^\star$ and $t^\star$ corresponding to the left and right boundary of the outer face, and each edge $e\in E$ induces a dual arc $e^\star$ directed from the vertex corresponding to the face of $G$ to the left of $e$ to the vertex corresponding to the face of $G$ to the right of $e$. We impose lower and upper bounds  on the flow $X(e^\star)$ for every dual arc, so that $1\le X(e^\star)\le \sigma(e)$. Finally, we complete $\cal N$ with an edge directed from $t^*$ to $s^*$, and assign it trivial upper and lower capacity bounds (i.e., $1\leq X((t^*,s^*))\leq |E| \cdot \max_{e \in E} \sigma(e)$).  
    
    We show that $G$ admits an embedding-preserving upward-planar layered drawing such that the span of each edge $e \in E$ is at most $\sigma(e)$ if and only if $\cal N$ admits a feasible circulation~$X$. 

    If such a drawing $\Gamma$ exists, then we set $X(e^*)$ to be the span of $e$, for each edge $e$ of $G$, and we set $X((t^*,s^*))$ to be the sum of the flows assigned to the arcs of $\cal N$ incoming into $t^*$ (outgoing from $s^*$). Clearly, such an assignment satisfies the capacity constraints and the flow conservation property at $s^*$ and $t^*$. Since $G$ is an st-graph, every internal face $f$ is bounded by exactly two directed paths from its face-source $s_f$ to its face-sink $t_f$. 
    Along either paths, the sum of the spans of their edges equals the difference between the $y$-coordinate of $t_f$ and the $y$-coordinate of $s_f$. Therefore, this sum coincides for the left and the right path of $f$. By construction, this equality yields the flow conservation at the vertex of $\cal N$ corresponding~to~$f$. 

    Conversely, given a feasible circulation $X$, we fix $y(s)=0$ and define $y(v)$, for every vertex $v \in  V\setminus \{s\}$, as the sum of $X(e^*)$ over all the edges $e$ belonging to any directed $s$-$v$-path. This value is well defined since any two such paths differ by substitutions along face boundaries, and feasibility implies that the two boundary paths of every internal face have equal total span.  Clearly, the span of each edge $e$ of $G$ coincides with $X(e^*)$, and thus it is bounded by $\sigma(e)$.
    Let $P_0$ be the left boundary of the outer face of $G$. Since $G$ is a plane st-graph, it is possible 
    to
    obtain $G$ starting from $P_0$, by repeatedly adding in the right outer face
    a directed path $P_i$, such that $P_i$ is the right path of an internal face of $G$, see, e.g., \cite{AngeliniLBF17,BattistaT88,DBLP:journals/tcs/FratiGW14,Mel}. 
    We compute a layered drawing $\Gamma$ of $G$ as follows. We first draw $P_0$ as a $y$-monotone curve with each vertex placed at its assigned $y$-coordinate. Then, 
    we draw each path $P_i$, with $i>0$, as a $y$-monotone curve that lies in the right outer face of the drawing of $\bigcup^{i-1}_{j=0} P_j$ with each vertex of $P_i$ placed at its assigned $y$-coordinate. By construction, $\Gamma$ is an embedding-preserving upward-planar layered drawing of $G$ with span at most $\sigma(e)$.
    
    Feasibility of a circulation with lower and upper capacity bounds can be reduced in linear time to the problem of finding a maximum $s$-$t$-flow~\cite{DBLP:books/daglib/0015106}, which can be found in $O(n)$ time since resulting network is planar and $s$ and $t$ are incident to the same face~\cite{hassin:81,henzinger_etal:97}.
\end{proof}

By setting $\sigma(e) = k$ for each edge $e$, \cref{th:st-plane} implies the following.

\begin{theorem}\label{cor:st-planE}
$k$-\textsc{span upward planarity} can be solved in linear time for plane st-graphs.
\end{theorem}

To extend this result to the setting in which the st-planar embedding is not prescribed, we need additional definitions. Let $G$ be an st-graph graph and let $u$ and $v$ be two vertices of $G$. Then, $\{u,v\}$ is a \emph{split pair} if $(u,v)$ is an edge of $G$ or 
$\{u,v\}$ is a \emph{separation pair} of $G$, i.e.,
$G-\{u,v\}$ consists of multiple connected components $C_1,C_2,\dots,C_k$, with $k\geq 2$ . 
A \emph{split component} of $G$ with respect to a split pair $\{u,v\}$ is either the edge $(u,v)$, if any, or 
the subgraph induced by $V(C_i) \cup \{u,v\}$, minus the edge $(u,v)$ if any, for $i = 1,\dots,k$.
A \emph{component} of $G$ with respect to a split pair $\{u,v\}$ is the union of any number of split components of $G$ with respect to $\{u,v\}$.

\begin{theorem}
\label{thm:st-planar}
$k$-\textsc{span upward planarity} can be solved in linear time for planar st-graphs.
\end{theorem}

\begin{proof} 
We show that a planar st-graph~$G$ has span at most~$k$ if and only if it admits an upward-planar layered drawing with span at most~$k$ for {\em any} of its st-planar embeddings. 
In other words, we show that prescribing the st-planar embedding of~$G$ does not affect the span of~$G$.
Together with \cref{th:st-plane}, this implies the  statement.

Obviously, if~$G$ admits an upward-planar layered drawing with span at most~$k$ for any of its st-planar embeddings, it has span at most~$k$. For the other direction, let~$\Gamma$ be an upward-planar layered drawing of~$G$ with span at most~$k$ and let~$\cal E$ be the planar embedding of~$G$. Clearly,~$\cal E$ is st-planar, since $s$ and $t$ are bottommost and topmost vertices in $\Gamma$, and so they lie in the outer face. Let~$G_{\cal E}$
be a plane st-graph obtained by prescribing the embedding of~$G$ to be~$\cal E$. Also, let~${\cal E'}$ be an st-planar embedding of~$G$ different from~$\cal E$ and let ~$G_{\cal E'}$ be the corresponding plane st-graph.
In the following, we show that is possible to iteratively modify~$\Gamma$ to transform it into an upward-planar layered drawing of~$G_{\cal E'}$ with span at most~$k$; refer to \cref{fig:st-flips}.
To this aim, we exploit the fact that~$\cal E$ can be transformed into~$\cal E'$ by only applying a finite sequence of flips of components of~$G$~\cite{AngeliniCBP13}. Each of these components~$D_i$ is defined with respect to a split pair~$\{u,v\}$ and has the property that the edges of~$D_i$ incident to~$u$ (resp.\ to~$v$) appear consecutively around~$u$ (resp.\ around~$v$). In other words,~$D_i$ is either a split component with respect to~$\{u,v\}$ or the union of any number of split components that appear consecutively in a parallel composition with respect to~$\{u,v\}$. By \cite{BattistaT96},~$D_i$ is a planar st-graph, with~$s,t \in \{u,v\}$ and~$s\neq t$.
Let~$\phi_1,\phi_2,\dots, \phi_r$ be the sequence of flips of components~$D_1,D_2,\dots,D_r$ that turn~$\mathcal E$ into~$\mathcal E'$. Also, let~$\Gamma_0 = \Gamma$ and suppose that~$\Gamma_{i-1}$ is an upward-planar layered drawing of~$G$ with span at most~$k$.
Let~$\Gamma_{i}$ be obtained from~$\Gamma_{i-1}$ as follows. Since~$D_i$ is st-planar, the drawing of~$D_i$ in~$\Gamma_{i-1}$ lies in the interior of a compact region~$R_i$ delimited by two y-monotone curves connecting~$u$ and~$v$ that (i) does not cross any edge and (ii) contains all and only the vertices of~$D_i$. Note that, the drawing $\Gamma'_i$ of~$D_i$ in~$R_i$ is a level-planar drawing of~$D_i$. Therefore, by horizontally mirroring $\Gamma'_i$ we obtain a level-planar drawing $\Gamma^*_i$ of $D_i$ with the same span as $\Gamma'_i$. To obtain $\Gamma_i$ from $\Gamma_{i-1}$, we replace $\Gamma'_i$ in the interior of $R_i$ with a a level-planar drawing of $D_i$ constructed by assigning to each vertex its y-coordinate as in $\Gamma'_i$ (and in $\Gamma^*_i$) and its x-coordinate so that the left-to-right order of the vertices along each y-coordinate is the same as in $\Gamma^*_i$, and by drawing the edges of $D_i$ as y-monotone curves such that the order in which such edges intersect any horizontal line is the same as in $\Gamma^*_i$. Clearly, the drawing of $D_i$ in $\Gamma_i$ is upward-planar and the span of the edges of $D_i$ is the same in $\Gamma_i$ as in $\Gamma'_i$. Therefore, we have that $\Gamma_i$ satisfies the required properties. Eventually, the final drawing $\Gamma_r$ is an upward-planar layered drawing of~$G$ with span at most~$k$ that preserves~$\cal E'$. 
\end{proof}

\begin{figure}
    \centering
    \begin{subfigure}[b]{.45\textwidth}
    \centering
    \includegraphics[page=11]{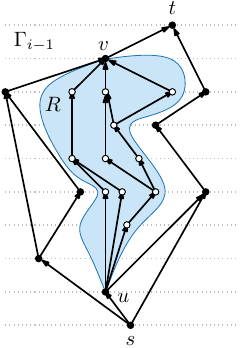}
    \subcaption{}\label{fig:st-before}
    \end{subfigure}
    \hfil
    \begin{subfigure}[b]{.45\textwidth}
    \centering
    \includegraphics[page=12]{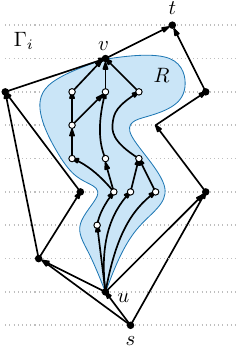}
    \subcaption{}\label{fig:st-after}
    \end{subfigure}
    \caption{
    Obtaining an upward-planar layered drawing $\Gamma_i$ (b) from an upward-planar layered drawing $\Gamma_{i-1}$ (a) by redrawing the uv-component defined by the split pair $\{u,v\}$ in the interior of a region $R$ containing the drawing of such a component in $\Gamma_i$.
    }
    \label{fig:st-flips}
\end{figure}

Next, we study upward-planar graphs that have a fixed upward-planar embedding and that still have a single source, but possibly multiple sinks.

\begin{lemma}
    \label{le:single-source-upward}
    For an $n$-vertex single-source upward-plane graph $G$ and any function \mbox{$\sigma: E \rightarrow \mathbb{N}$}, we can decide in $O(n)$ time whether $G$ has an embedding-preserving upward-planar layered drawing such that the span of each edge $e \in E$ is at most $\sigma(e)$.
\end{lemma}

\begin{proof}
Let $G$ be a single-source upward-plane graph. 
The authors of~\cite{BertolazziBMT98} show that, in any embedding-preserving upward-planar drawing $\Gamma$ of $G$, the following holds: 
\begin{itemize}
\item For each internal face $f$, at most one sink-switch (the topmost vertex of $f$ in $\Gamma$) is not a sink of $G$ and all but one sink-switches form a large angle in $f$.
\item For the outer face $h$ of $\Gamma$, all the sink-switches are sinks of $G$ and form a large angle in $h$.
\end{itemize}

\noindent
Furthermore, as a consequence of the above, they show that $G$ can be augmented to a plane st-graph $G'$ as follows:
\begin{itemize}
\item For each internal face $f$, let $t_f$ be the unique sink-switch of $f$ that is not assigned a large angle in $f$. The graph $G'$ contains an edge directed from each sink-switch that is assigned a small angle in $f$ to $t_f$.
\item The graph $G'$ contains a sink $t$ and an edge directed from each sink of $h$ to $t$. 
\end{itemize}

We show that $G$ has an embedding-preserving upward-planar layered drawing such that the span of each edge $e \in E$ is at most $\sigma(e)$ if and only if $G'$ has an embedding-preserving upward-planar layered drawing such that the span of each edge $e$ is at most $\sigma(e)$ if $e$ is also in $G$ and is at most $n-1$ otherwise.
One direction is obvious, since $G \subseteq G'$. 
For the other direction, suppose that $G$ has such a drawing $\Gamma$. We can assume w.l.o.g.\ that there is no $y$-coordinate between the lowest and the largest vertex that does not contain a vertex; otherwise, if there is no vertex on $y$-coordinate $y'$, then one can easily compress the part of the drawing between $y$-coordinates $y'-1$ and $y'+1$ to height 1, effectively reducing the height of the drawing by~1. Thus, the height of $\Gamma$ is at most $n$ (since introducing $t$ in the outer face increased the graph size by~1) and there is no edge with span greater than $n$.

We show that $\Gamma$ can be augmented to an embedding-preserving upward-planar layered drawing $\Gamma'$ of $G'$ respecting the given bound on the span of each edge. To this aim, observe first that, as shown in~\cite{BertolazziBMT98}, for each internal face $f$ of $G$, in $\Gamma$, we can connect each sink-switch of $f$ that forms a large angle in $f$ with the unique sink-switch of $f$ that forms a small angle in $f$ with a crossing-free $y$-monotone curve; refer to \cref{fig:singleSource-inner}. 
Then, since all the sinks of $G$ in the outer face of $\Gamma$ form a large angle in this face, we can place the vertex $t$ (in the outer face of $\Gamma$) one unit above the top-most sink of $G$ and connect
all the sinks of $G$ in the outer face of $\Gamma$ 
with a crossing-free $y$-monotone curve to $t$; refer to \cref{fig:singleSource-outer}. This yields an embedding-preserving upward-planar layered drawing $\Gamma'$ of $G'$. Clearly, in $\Gamma'$, all the edges of $G'$ that belong to $G$ have span at most $\sigma(e)$, by hypothesis, and all the edges of $G'$ that do not belong to $G'$ have span strictly smaller than the height of~$\Gamma$, which is at most $n$. 

Since the construction of $G'$ takes $O(n)$ time, the lemma then follows by \cref{th:st-plane}.
\end{proof}

By setting $\sigma(e) = k$ for each edge $e$, \cref{le:single-source-upward} implies the following.

\begin{theorem}\label{th:single-source-upward}
$k$-\textsc{span upward planarity} can be solved in linear time for single-source upward-plane graphs.
\end{theorem}

Note that \cref{th:single-source-upward} can be extended to an XP algorithm for upward-plane graphs parameterized by the number of sources. We refrain from doing this here, as we will give an XP algorithm for a more general setting later in \cref{cor:xp}.

We now deal with graphs for which the planar embedding is fixed (but the upward-planar embedding is not). The next lemma is implicit in~\cite{BertolazziBMT98}, and is instrumental for the proof of the subsequent results.

\begin{lemma}[\cite{BertolazziBMT98}]
\label{le:single-source-unique-up-embedding}
For a single-source plane graph $G=(V,E)$, the upward-planar embedding---if it exists---is unique up to reflection and can be computed in $O(n)$ time.
\end{lemma}

\begin{lemma}\label{le:single-source-plane}
    For an $n$-vertex single-source plane graph $G=(V,E)$ and any function $\sigma: E \rightarrow \mathbb{N}$, we can decide in $O(n)$ time whether $G$ has an embedding-preserving upward-planar layered drawing such that the span of each edge $e \in E$ is at most $\sigma(e)$.
\end{lemma}
\begin{proof}
If $G$ has no embedding-preserving upward-planar drawing, we reject the instance. Otherwise, 
by \cref{le:single-source-unique-up-embedding}, the given planar embedding of $G$ determines its upward-planar embedding. Thus, by \cref{le:single-source-upward}, we can decide in $O(n)$ time whether $G$ has an embedding-preserving upward-planar layered drawing where each edge $e$ has span~at~most~$\sigma(e)$.\end{proof}

By setting $\sigma(e) = k$ for each edge $e$, \cref{le:single-source-plane} implies the following.

\begin{theorem}\label{th:single-source-plane}
$k$-\textsc{span upward planarity} can be solved in linear time for single-source plane graphs.
\end{theorem}

Next, we show that, for graphs with a fixed planar embedding, the problem is in XP when parameterized by the number of sources.

\begin{theorem}
\label{cor:xp}
$k$-\textsc{span upward planarity} can be solved in %
$O(2^{z}n^{2z+1})$
time for $n$-vertex plane graphs with $z$ sources.
\end{theorem}

\begin{proof}
	Let $S$ be the set of sources of $G$ and assume that the planar embedding of $G$ is fixed.
	We reduce the instance to at most $2^zn^{2z}$ instances of the single-source case with fixed plane embedding.
	We introduce a new vertex $r$ (the \emph{super-source}) embedded in the outer face.
	While the graph still contains a source~$s$ different from~$r$, we branch as follows.
	First, we pick a \emph{corner} of~$s$, i.e., a pair $\langle e_1,e_2 \rangle$ of consecutive incident edges of $s$ (note that $e_1=e_2$ is possible if $\deg(s)=1$). Let $f_s$ be the \emph{associated face} of the corner $\langle e_1,e_2 \rangle$, i.e., the face between $e_1$ and $e_2$.  %
	Pick another corner $\langle e'_1,e'_2 \rangle$ with associated face $f_s$ and let $v_s$ be the common incident vertex of $e_1'$ and $e_2'$. 
	Add the edge $(v_s,s)$ and route it inside $f_s$ from $v_s$ between $e_1'$ and $e_2'$ to $s$ between $e_1$ and $e_2$. If this introduces loops, parallel edges, or directed cycles, we immediately reject the branch.  Otherwise, we continue branching until~$r$ is the only source.
	Let $G'=(V\cup\{r\},E')$ be the resulting graph.
	By construction, $G'$ is plane, acyclic, and has a single source $r$. We now use \cref{le:single-source-plane} to decide whether $G'$ has an embedding-preserving upward-planar layered drawing such that the span of each edge $e\in E'$ is at most $\sigma(e)$, where $\sigma(e)=k$ if $e\in E$, and $\sigma(e)=n$
	if $e\in E'\setminus E$. See \cref{fig:xp*} for an illustration.
	
	\begin{figure}[t]
		\centering
		\includegraphics[page=8]{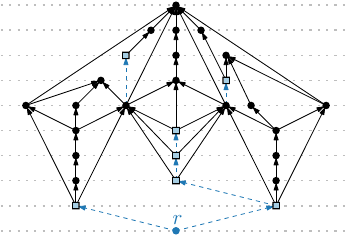}
		\caption{Illustration for the proof of \cref{cor:xp}. Sources are drawn as squares, added edges are drawn dashed.}
		\label{fig:xp*}
	\end{figure}
	
    For each source~$s$ and each of the $\deg(s) \le n$ incidences of $s$ to a face $f_s$, we branch over $|\partial f_s|$ choices 
    for adding an incoming edge to $s$ inside $f_s$, where $|\partial f_s|$ denotes the number of corners (i.e., vertex occurrences, not distinct vertices) along the boundary of $f_s$.
    
	We prove by induction on the number of edges of the current graph that $|\partial f_s| \leq 2n$ for $n\geq 1$; recall that the graph consists of $G$ plus $r$ and, thus, has $n+1$ vertices. If~$f_s$ is the only face, then the current graph is a forest and $|\partial f_s| \leq 2n$ by the Handshaking Lemma. Assume now that $e$ is an edge incident to $f_s$ and another face. We apply the inductive hypothesis to the graph $G''$ obtained after removing $e$. Let $f$ be the face of $G''$ containing $f_s$. Then $|\partial f_s| \le |\partial f| \leq 2n$.
	
	Thus, in total, we obtain at most $2n^2$ choices for each source and, hence, in total $(2n^{2})^z$ branches.  
	For each branch, we invoke \cref{le:single-source-plane}, which requires $O(n)$ time.
	This yields an overall running time of $O(2^{z} n^{2z+1})$.

We now show that $G$ has span at most $k$ if and only if at least one branch leads to a yes-instance. One direction is obvious, since $G \subseteq G'$. 
For the other direction, suppose that $G$ admits an embedding-preserving drawing $\Gamma$ of span at most $k$.
As in the proof of \cref{le:single-source-upward}, we can assume w.l.o.g\ that $\Gamma$ has height at most $n-1$ and thus every edge has span at most $n-1$.
Insert the super-source $r$ on the outer face of $\Gamma$ one level below the lowest vertex in $\Gamma$.

Consider any source $s\in S$. There is a unique face $f_s$ that has a large angle ($>\pi$) at $s$. If $f_s$ is an internal face, then there must be at least one vertex $v_s$ below $s$ that is visible from $s$ in $f_s$, that is, we can add the directed edge $(v_s,s)$ in $f_s$ as a straight-line segment without crossings. 
If $f_s$ is the outer face, there are two possibilities.
If $s$ is not one of the lowest vertices among the sources on the outer face in $\Gamma$, then we can find a vertex $v_s$ on $f_s$ such that we can add the straight-line edge $(v_s,s)$ with span at most $n$ as in the previous case. Otherwise, we have $y(s)-y(r)=1$, so we can add $(r,s)$ as a straight-line segment of span 1 on the outer face.

In all cases, the added edge can be drawn entirely within the chosen face without crossings and with span at most $n$.
Therefore, we obtain a drawing of the resulting augmented graph $G'$ where each edge $e$ has span at most $\sigma(e)$.
Hence, the branch that guesses exactly these choices is a yes-instance.
Therefore, $G$ has span at most $k$ if and only if at least one of the at most 
single-source instances is a yes-instance.
\end{proof}

\begin{figure}[t]
    \centering
    \includegraphics[page=8]{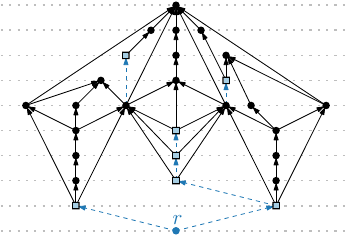}
    \caption{Illustration for the proof of \cref{cor:xp}. %
    Added edges are drawn dashed.}
    \label{fig:xp}
\end{figure}

Finally, we turn our attention to single-source outerplanar graphs.

\begin{theorem}
\label{th:outerplanar}
$k$-\textsc{span upward planarity} can be solved in linear time for 
single-source outerplanar graphs.
\end{theorem}

\begin{proof}
Let $s$ be the unique source of $G$.
Compute an outerplanar embedding $\mathcal{E}$ of $G$ in $O(n)$ time.
We claim that $G$ has an upward-planar layered drawing of span at most $k$
if and only if it has an embedding-preserving such drawing with respect to $\mathcal{E}$.
Assume first that $G$ is biconnected, so it has a unique outerplanar embedding.
The if-direction is immediate.

For the other direction, suppose that $G$ admits an upward-planar layered drawing $\Gamma$ of span at most $k$
with a non-outerplanar embedding $\mathcal{E}'$.
We transform $\mathcal{E}'$ into $\mathcal{E}$ by repeatedly applying the following flip operation; see ; see \cref{fig:outerplanar-splitpair}.
Let $e=(u,v)$ be an edge that is \emph{internal} (a chord) in $\mathcal{E}$ but lies on the outer face in the current embedding.
Since $G$ is biconnected and outerplanar, $e$ is a separation pair: the graph $G-\{u,v\}$ has exactly two connected components,
say $A$ and $B$. Assume w.l.o.g.\ that $s\in A$. We now \emph{flip} $B$ along $e$, maintaining planarity and the $y$-coordinates
of all vertices, thus preserving the span of every edge: we mirror $B$ and move it sufficiently far to the other side of $e$ such that it does not intersect any part of $A$. Since $G$ is biconnected and has a single source, $u$ is the single source of $G[B\cup\{u,v\}]$, so the vertices in $B$ cannot have any $y$-coordinate lower than $y(u)$. Assume that $e$ is drawn such that the outer face lies to the right of $e$ and consider $G[A\cup\{u,v\}]$. If we walk on the outer facial cycle of $G[A\cup\{u,v\}]$ in counter-clockwise direction, starting from $v$, until we reach (one of the) vertices in the drawing with the largest $y$-coordinate, and any edge that goes downward must be directed to the left, as otherwise a global source would be required to ``turn around''. Thus, no part of $B$ can ``stick out'' to the bottom, and no part of $A$ can ``stick in'' from the top to the area that we want to draw $B$ in after flipping it. Thus, there is enough space to place the vertices of $B$ on the other side of $e$ such that we can still draw its edges as $y$-monotone curves without intersection. In particular, this can be achieved by moving the vertices of $B$ that lie on the $y$-coordinates between $y(u)$ and $y(v)$ to be very close to $e$, while moving the vertices that lie above $v$ very far away.%

Each flip strictly decreases the number of edges that are internal in $\mathcal{E}$ but currently lie on the outer face.
Hence, after finitely many flips we obtain an embedding that coincides with $\mathcal{E}$,
and we have transformed $\Gamma$ into an embedding-preserving drawing with respect to $\mathcal{E}$ of the same span.
Since the outer face is fixed, there is a unique upward-planar embedding~\cite{BertolazziBMT98}, so we can
use \cref{th:single-source-upward} to decide whether $G$ has span at most $k$ in $O(n)$ time.

Assume now that $G$ is not biconnected and consider the block-cut tree of $G$.
Choose an arbitrary block $B_0$ that contains the global source $s$ as the root block.
For every other block $B\neq B_0$, let $c(B)$ be the unique cut vertex through which $B$ is connected to its parent block.
Then $c(B)$ is the unique source in the subgraph induced by $B$; otherwise, there would be another global source in $B$ or in one of its child blocks.
Therefore, every block $B$ is a single-source biconnected outerplanar graph, so it has an upward-planar drawing with span at most $k$
if and only if it has an outerplanar one, which we can decide for each block independently in $O(n)$ time as outlined above.
If any block is a no-instance, then $G$ is a no-instance.

Assume now that every block admits an upward-planar layered drawing of span at most $k$.
We construct a drawing of $G$ by gluing the block drawings at their shared cutvertices.
For every cut vertex $x$, there is exactly one block $B_x$ that has $x$ as a local sink, and all other blocks incident to $x$ (children of $B_x$ in the block-cut tree) have $x$ as their (unique) local source. Hence, we can first place the drawing of $B_x$, and then the drawings of the other blocks incident to $x$ above $x$ next to each other in any order.
This preserves upward planarity and the span of each edge in every block; see \cref{fig:outerplanar-cutvertex}.
Hence, $G$ has span at most $k$ if and only if all of its blocks pass the test.
\end{proof}

\begin{figure}[h]
    \centering
    \subcaptionbox{\label{fig:outerplanar-splitpair}}{\includegraphics[page=9]{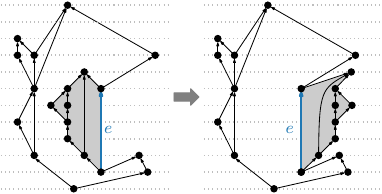}}
    \hfil
    \subcaptionbox{\label{fig:outerplanar-cutvertex}}{\includegraphics[page=10]{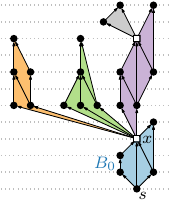}}
    \caption{Illustration for the proof of \cref{th:outerplanar}: (a) the flipping operation, and (b) laying out biconnected components around a cutvertex.}
    \label{fig:outerplanar}
\end{figure}

\section{Parameterized Algorithms}\label{sec:fpt}

We first recall some structural graph parameters.  Let~$G=(V,E)$ be a graph.
The \emph{vertex cover number} $\vc(G)$ is the size of a minimum \emph{vertex cover} of $G$, i.e., a set of vertices such that each edge is incident to at least one of them. The \emph{vertex integrity} $\vi(G)$ is the smallest number~$k$ such that there is a set~$M \subseteq V$ such that $|M|$ plus the size of the largest connected component of~$G-M$ is at most~$k$.  A \emph{treedepth decomposition} of $G$ is a rooted tree~$T=(V,E_T)$ on the vertex set of~$G$ such that for each edge in $E$ between vertices $u$ and $v$, we have that~$u$ is an ancestor of~$v$ in~$T$ or vice versa.  The \emph{treedepth} $\td(G)$ is the minimum height of any treedepth decomposition of~$G$.  Note that~$\td(G) \le \vi(G) \le \vc(G)$.

Bekos et al.~\cite{bekos_etal_span:gd24} establish a paradigm for constructing FPT
algorithms for computing the span of undirected graphs with respect to structural graph parameters.  To show that testing whether~$\spn(G) \le s$ is FPT with respect to a graph parameter $k$,
they establish (i) that the problem admits a kernel with respect
to~$k+s$ and (ii) that $s$ can be bounded by a function of~$k$.
Together, this yields a kernel whose size is bounded by a function
of~$k$ alone.  Bekos et al.~\cite{bekos_etal_span:gd24} apply this method to obtain FPT algorithms for the span of undirected graphs parameterized by the vertex cover number, the vertex integrity, and the treedepth.
We show in \Cref{sec:fpt-algor-treed} that their techniques for step (i) can be extended to the directed case.%

\subsection{Parameterization by Treedepth plus Span}
\label{sec:fpt-algor-treed}

We start by recalling some definition from~\cite{bekos_etal_span:arxiv}.
Let~$G=(V,E)$ be a graph and let~$M \subseteq V$ be a subset of~$k$
vertices.  Consider a connected component~$C$ of~$G-M$.  A vertex
of~$M$ that is adjacent to a vertex in~$C$ is called an
\emph{attachment of~$C$}.  We use~$\att(C)$ to denote the set of
attachments of~$C$.  The $M$-bridge~$B$ of~$C$ is the subgraph formed by
the vertices in~$V(C) \cup \att(C)$ together with all edges that have
at least one endpoint in~$C$.  We also refer to~$\att(C)$ as the
attachment of~$B$.

Planarity implies that $G-M$ has at most~$2k-4$ $M$-bridges with
three or more attachments and at most~$3k-6$ pairs~$\{u,v\} \subseteq M$ such that
there is an $M$-bridge with attachments~$u$ and~$v$~\cite[Lemma 10]{bekos_etal_span:arxiv}.  
Two bridges~$B$ and~$B'$ are \emph{equivalent} if there is an isomorphism between them
that leaves the attachments fixed.  We note that, while Bekos et
al. consider an isomorphism between undirected graphs, we require an
isomorphism as directed graphs.  In particular, two components that are equivalent according to our definition are also equivalent in the undirected setting and hence the results from~\cite{bekos_etal_span:arxiv} apply.

Bekos et al.~\cite[Lemma 23]{bekos_etal_span:arxiv} show that if there are more than~$(4s+4)b$ equivalent components of size~$b$ with a single
attachment~$v$, then reducing their number to $(4s+4)b$ yields an
equivalent instance.  Similarly, if there are more than~$(8s+8)(b+1)$
equivalent components of size~$b$ with two attachments, then reducing their number
to~$(8s+8)(b+1)$ yields an equivalent instance.  The proof is based on
the fact that these numbers guarantee that, in any layered planar drawing
with span at most~$s$, one of these components is drawn entirely above
or entirely below its attachment (in the case of a single attachment)
or strictly between its two attachments (in case of two attachments),
and thus an arbitrary number of additional isomorphic copies can be
reinserted next to it in a drawing.  It is readily seen that such a
reinsertion of copies is also possible in the upward-planar layered setting, and thus these
reduction rules are also safe in the directed case.
We are now ready to prove \Cref{thm:fpt-para+span}.

\begin{theorem}
  \label{thm:fpt-para+span}  
  The problem of computing the span of a DAG $G$ admits a kernel with respect to~$s + k$, where $s$ is the span and $k = \vc(G)$, $k=\vi(G)$, or~$k=\td(G)$. In case of the vertex cover number, the size of the kernel is $O(s \cdot k)$.
\end{theorem}

\begin{proof}
We note that applying the above reductions with~$b = |M| = k$ immediately yields a kernel with respect to vertex-integrity plus span.  Moreover, Bekos et al.~\cite[Section 4.3]{bekos_etal_span:arxiv} show that, with suitably chosen sets~$M$ and
sizes~$b$, these reductions can be used to process a treedepth
decomposition in a bottom-up fashion so that the degree of the nodes
in the decomposition trees becomes bounded by a function of the
treedepth $\td$ and the span $s$.  As argued above, the reductions also
yield an equivalent instance in the directed setting.  Thus, the
kernelization algorithm of Bekos et al. also works in the directed
setting and their analysis of the size of the resulting kernel
applies. 

In the special case where $M$ is a vertex cover of size~$k$, every
component of~$G-M$ has size~$b=1$ and any two components, which
consist of a single vertex, are equivalent if and only if they have
the same attachments and the edges connecting to these attachments
have the same directions.  Thus, the number of equivalence classes for vertices of degree at most~$2$ with the same attachments is bounded by a constant, and we obtain a kernel of size~$O(k \cdot s)$.
\end{proof}

\subsection{Parameterization by Vertex Cover}
\label{sec:fpt-vc}

As mentioned above, Bekos et al. in step (ii) of their paradigm extend their FPT algorithms to depend on the structural parameter $k$ only by bounding the span~$s$ in terms $k$.  We show that, in the directed case, it is not possible to remove the dependency on $s$ in a similar fashion.  Namely, the following observation shows that their step (ii) fails for directed graphs.

\begin{observation}
For every $s\in \mathbb N$, there is a DAG $G_s$ with $\vc(G_s) = 2$ and $\spn(G_s) \ge s$.
\end{observation}

\begin{proof}
  Let~$G_s$ be $K_{2,2s}$ with all edges directed towards the vertices of degree~$2$; see \Cref{fig:small_cover_span}. Let~$u,v$ denote the two vertices of degree~$2s$ and $W$ denote the set of vertices of degree~$2$.  Note that~$\{u,v\}$ is a vertex cover of size~$2$. In an upward-planar layered drawing of~$G_s$, there are at most two vertices of~$W$ on each layer; see \Cref{fig:upward-g3}.  Thus, the height of the drawing is at least~$s$.  Since~$u$ lies below all vertices of~$W$, one of its edges has span at least~$s$.
\end{proof}

\begin{figure}
    \centering
    \begin{subfigure}[b]{.4\textwidth}
    \centering
    \includegraphics[page=1]{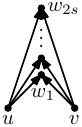}
    \subcaption{}\label{fig:small_cover_span}
    \end{subfigure}
    \hfil
    \begin{subfigure}[b]{.4\textwidth}
    \centering
    \includegraphics[page=2]{figures/small_cover_span}
    \subcaption{}\label{fig:upward-g3}
    \end{subfigure}
    \caption{(a) Graph $G_s$ and  (b) an upward drawing of~$G_s$.}
    \label{fig:span-family}
\end{figure}

Nevertheless, we show in the following that computing the span of a directed graph is FPT with respect to the vertex cover number alone.  To this end, let~$G=(V,E)$ be a planar directed acyclic graph, let~$M$ be a vertex cover of size~$k$ in~$G$ and consider the problem of testing whether~$G$ has span at most~$s$.

Again, planarity implies that there are at most~$2k$ vertices in $G-M$ with degree at
least~$3$.  We thus focus on the vertices of~$G-M$ with degree at
most~$2$.  It is readily seen that for each vertex~$v \in M$, it suffices to keep
one upper and lower degree-1 neighbor of~$v$ (if it even has one) and to remove the remaining ones.
Thus, we can assume that there are at most~$2k$ vertices of degree~$1$
in~$G-M$.  Similarly, it suffices to keep one transversal vertex~$w$
with neighbors~$u$ and~$v$ and to remove the remaining ones; this limits their number to~$3k$ by
\cite[Lemma 10]{bekos_etal_span:arxiv}.  The issue is that for a pair of vertices~$u,v \in M$ there
may be an unbounded number of upper or lower $uv$-ears.

Consider an upward-planar layered drawing~$\Gamma$ of~$G$ with
corresponding leveling~$\ell$.  Let~$w$ be a $uv$-ear, where~$\ell(u) \le \ell(v)$ and let~$\gamma$ be a~$y$-monotone curve from~$u$ to~$v$ that does not cross any edge incident to~$w$
in~$\Gamma$.  Consider the closed curve~$\rho$ formed by~$\gamma$
together with the two edges incident to~$w$ in~$\Gamma$.  If the region that lies
to the right of~$\rho$ when traversing it in clockwise direction is
bounded, we call the ear~$w$ \emph{clockwise}; otherwise, the region on the
other side is bounded, and we call $w$ \emph{counterclockwise}; see \Cref{fig:cw_ccw_ears}.  Note
that these notions only depend on~$\Gamma$ and not on the choice of the curve~$\gamma$.

We call a set~$X$ of clockwise upper $\{u,v\}$-ears \emph{consecutive}
in~$\Gamma$ if for any pair of distinct vertices~$w_1,w_2\in X$, the
set of vertices enclosed by the 4-cycle formed by~$uw_1vw_2$ is a
subset of~$X$; see \Cref{fig:consecutive-ears} for an illustration.  We define consecutivity in an analogous fashion for
counterclockwise upper ears and for (clockwise/counterclockwise) lower
ears.

\begin{figure}
    \begin{subfigure}[b]{.5\textwidth}
    \centering
    \includegraphics[page=1]{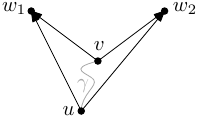}
    \subcaption{}
    \label{fig:cw_ccw_ears}
    \end{subfigure}
    \begin{subfigure}[b]{.5\textwidth}
    \centering
    \includegraphics[page=2]{figures/cw_ccw_ears}
    \subcaption{}
    \label{fig:consecutive-ears}
    \end{subfigure}
    \caption{(a) Two upper~$uv$-ears~$w_1$ and~$w_2$.  The ear~$w_1$ is counterclockwise and~$w_2$ is clockwise. (b) Clockwise upper $uv$-ears~$x,y,z$.  The set~$\{x,y\}$ is consecutive, but~$\{x,y,z\}$ is not due to the white vertices.}
\end{figure}

\begin{lemma}
  \label{lem:ears-consecutive}
  Let~$G=(V,E)$ be a DAG with a vertex cover $M$ of
  size~$k$ that admits an upward-planar layered
  drawing~$\Gamma$ with span~$s$.  Then~$G$ admits an upward-planar layered drawing~$\Gamma'$ with span~$s$ where for each
  pair of vertices~$u,v \in M$ all upper/lower
  clockwise/counterclockwise~$\{u,v\}$-ears are consecutive.
\end{lemma}

\begin{proof}
    Let~$X$ be the set of clockwise upper~$uv$-ears and assume that~$X$ is not consecutive in~$\Gamma$.  Let~$\ell$ denote the corresponding leveling and without loss of generality assume~$\ell(u) \le \ell(v)$.
    Then there exist clockwise upper ears~$w_1,w_2$ with~$\ell(w_1) < \ell(w_2)$ so that the cycle~$uw_1vw_2$ encloses at least one vertex that does not lie in~$X$.  By choosing~$w_1$ and~$w_2$ so that~$\ell(w_2) - \ell(w_1)$ is minimized, we may further assume that the set~$Y$ of vertices enclosed by the cycle is disjoint from~$X$; see \Cref{fig:non-consec-ears}.  Observe that, since $u$ and~$v$ are the only neighbors of~$w_1$ and~$w_2$, every edge that connects a vertex of~$Y$ to a vertex outside of~$Y$ ends either at~$u$ or at~$v$.  Moreover, $\ell(u) < \ell(y) < \ell(w_2)$ for all~$y \in Y$.
    
    We obtain a new upward-planar layered drawing by (i) shifting all vertices~$Y$ one level up and (ii) moving~$w_2$ to the level~$\ell(w_1)+1$ and changing the embedding so that the edge~$uw_2$ immediately succeeds the edge~$uw_1$ in the clockwise order of the edges around~$u$ and, likewise, the edge~$vw_2$ immediately precedes the edge~$vw_1$ in the clockwise order of the edges around~$v$; see \Cref{fig:make-consec}.  Let~$\Gamma'$ denote the new drawing and let~$\ell'$ denote the new leveling.  
    
    Note that~$\Gamma'$ is again an upward-planar layered drawing; the edges within~$Y$ can be moved up rigidly, the fact hat~$w_2$ does not have neighbors other than~$u$ and~$v$ guarantee sufficient space on level~$\ell(w_2)$ to position the vertices of~$y \in Y$ with~$\ell(y) = \ell(w_2)-1$, and edges connecting~$uy$ or~$vy$ can be extended without introducing crossings.  Moreover, the edges~$uw_2$ and~$vw_2$ can be drawn along the edges~$uw_1$ and~$vw_1$ and then on to the level above, where there is space for~$w_2$ between the two vertices or edges left and right of~$w_1$ on level~$\ell(w_1)$ in~$\Gamma$.

    Moreover, we claim that~$\spn(\Gamma') \le \spn(\Gamma)$.  To see this, observe that the only edges whose span changes are edges incident to~$u$ or~$v$.  The edges incident to~$w_2$ even decrease their span.  Thus, if the span increases, it must be due to an edge that connects~$u$ or~$v$ to an vertex of~$Y$.  Let~$xy$ with~$x \in \{u,v\}$ and~$y \in Y$ be an edge.  Then~$\spn_{\Gamma'}(xy) = \ell'(y) - \ell'(x) = \ell(y)+1 - \ell(x) \le \ell(w_2) - \ell(u) = \spn_{\Gamma}(uw_2)$.   Thus~$\spn(\Gamma') \le \spn(\Gamma)$.

    By repeatedly applying this operation, we may assume that all upper clockwise $uv$-ears are consecutive in~$X$.  Note further that this does not invalidate the consecutivity of counterclockwise upper $uv$-ears or clockwise/counterclockwise lower $uv$-ears, or any kinds of $u'v'$-ears with~$\{u,v\} \ne \{u',v'\}$.  Moreover, with an analogous operation also counterclockwise upper $uv$-ears as well as clockwise and counterclockwise lower~$uv$-ears can be made consecutive. Applying this modification for each pair of vertices~$u,v \in M$ that has $uv$-ears thus yields the desired drawing.
\end{proof}

\begin{figure}
    \centering
    \begin{subfigure}[b]{.3\textwidth}
    \centering
    \includegraphics[page=1]{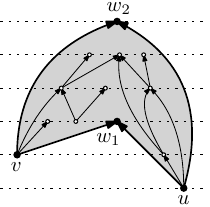}
    \subcaption{}\label{fig:non-consec-ears}
    \end{subfigure}
    \begin{subfigure}[b]{.3\textwidth}
    \centering
    \includegraphics[page=2]{figures/consecutive_ears}
    \subcaption{}\label{fig:make-consec}
    \end{subfigure}
    \caption{(a) An initial drawing, where the clockwise upper $uv$-ears~$w_1$ and~$w_2$ are not consecutive.  The area enclosed by them is shaded gray and the enclosed vertices are drawn as small white disks.  (b) A modified drawing, where~$w_1$ and~$w_2$ are consecutive, obtained by shifting the enclosed vertices one level up and moving~$w_2$ down to the level above~$w_1$.}
\end{figure}

It follows that we may assume that for each pair of vertices~$\{u,v\}$
the upper/lower clockwise/counterclockwise ears appear consecutively in an upward-planar layered drawing~$\Gamma$ of $G$.
For a pair~$u,v \in M$, let~$U_{u,v}$ and~$L_{u,v}$ denote the set of upper and ears in $G$, respectively.  If~$|U_{u,v}| > 6$,
let~$U_{u,v}' = \{w_1^{u,v},\dots ,w_6^{u,v}\}$ be a subset of six
vertices.  We then remove from~$G$ all vertices
of~$U_{u,v} \setminus U_{u,v}'$ and add the \emph{upper block replacement edges}~$u_1^{u,v} = w_1^{u,v}w_3^{u,v}$
and~$u_2^{u,v} = w_2^{u,v}w_4^{u,v}$.  Similarly, if~$|L_{u,v}| > 6$, we keep a set~$L'_{u,v} \subseteq L_{u,v}$, remove the remaining ears and add corresponding \emph{lower block replacement} edges~$l_1^{u,v}$ and~$l_2^{u,v}$.  Let~$G' = (V',E')$ be obtained
by applying these reductions to all pairs~$u,v \in M$ and let~$S \subseteq E'$ be
the set of (upper and lower) block replacement edges.  We call an upward-planar layered drawing of~$G'$ \emph{$s$-good} if (i) each edge in~$E' \setminus S$ has span at most~$s$ and (ii) for each pair~$u,v \in M$ for which $G'$ contains an upper block replacement edge, we have $\spn(u_1^{u,v}) + \spn(u_2^{u,v}) \ge |U_{u,v}| - 4$ and (ii) for each pair~$u,v \in M$ for which $G'$ contains a lower block replacement edge, we have $\spn(l_1^{u,v}) + \spn(l_2^{u,v}) \ge |L_{u,v}| - 4$.

\begin{lemma}
  \label{lem:ears-replaced}
  The graph $G$ has an upward-planar layered drawing with span~$s$ if and only if~$G'$ has an upward-planar layered drawing that is $s$-good.
\end{lemma}

\newcommand{\Xcw}{X_{\mathrm{cw}}}
\newcommand{\Xccw}{X_{\mathrm{ccw}}}

\begin{proof}
  Let~$\Gamma'$ be an $s$-good drawing of $G'$.  We show how to insert the removed upper $uv$-ears from~$U_{u,v} \setminus U'_{u,v}$ into~$\Gamma'$.  Doing this for all removed upper and lower ears for all pairs~$u,v \in M$ then yields the desired drawing~$\Gamma$ of~$G$. 
  
  Consider a block replacement edge $e=xy$ connecting two upper $uv$-ears~$x$ and~$y$ and let $\ell$ denote the lowest level that is crossed by $e$. Position a new vertex~$w$ by subdividing the edge~$e$ at the point where it crosses~$\ell$.  The edges~$uw$ and~$vw$ can be inserted as $y$-monotone non-crossing curves by first following~$ux$ and~$vx$, respectively, and then following the curve~$e$ until~$w$; see \Cref{fig:ear-reinsertion}.  Afterwards, we remove the edge from~$xw$.  This effectively reinserts the $uv$-ear~$w$ and replaces the block replacement edge~$xy$ by the upper block replacement edge~$wy$ whose span is one less than the span of~$xy$.  Note further that the span of~$uw$ is less than the span of~$uy$ and likewise the span of~$vw$ is less than the span of~$vy$.  Thus, this modification does not increase the span of the drawing.  By iterating this process with the shortened upper block replacement edge~$wy$, we can thus insert as many~$uv$-ears as levels are crossed by the initial upper block replacement edge, i.e.,~$\spn(xy)-1$ $uv$-ears.  The same can be done for the other upper block replacement edge of the pair~$u,v$.  Since~$\Gamma'$ is $s$-good, this suffices to remove all of the removed upper~$uv$-ears.  As mentioned above, we can independently perform this reinsertion for removed upper and lower ears and thus obtain the desired drawing~$\Gamma$ of~$G$.

  Conversely, let~$\Gamma$ be an upward-planar layered drawing of~$G$ with span~$s$.  By \Cref{lem:ears-consecutive}, we may assume that all the upper/lower clockwise/counterclockwise $uv$-ears are consecutive for each pair~$u,v \in M$.  We construct an $s$-good drawing~$\Gamma'$ of~$G'$ as follows.  Consider a pair~$u,v \in M$ that has more than six upper ears and let~$\Xcw$ and~$\Xccw$ denote the set of clockwise and counterclockwise upper $uv$-ears, respectively, and note that~$\Xcw$ and~$\Xccw$ are consecutive.  Since~$|\Xcw| + |\Xccw| \ge 7$, we can partition~$\Xcw \cup \Xccw$ into consecutive sets~$X_1,\dots,X_4$ with~$|X_1|=|X_2|=1$ and~$|X_3|,|X_4| \ge 2$.  As the upper $uv$-ears are completely exchangable, after renaming we may assume that~$X_1 = \{w_5^{u,v}\}$, $X_2 = \{w_6^{u,v}\}$, that the lowest vertices in~$X_3$ and~$X_4$ are~$w_1^{u,v}$ and~$w_2^{u,v}$, respectively, and that the highest vertices in~$X_3$ and~$X_4$ are~$w_3^{u,v}$ and~$w_4^{u,v}$, respectively; see \Cref{fig:ear-partition}.  Due to the consecutivity of~$X_3$ and~$X_4$, after removing all other $uv$-ears, we can insert the upper block replacement edges~$u_1^{u,v}$ and~$u_2^{u,v}$ connecting the lowest and the highest vertex of~$X_3$ and of~$X_4$, respectively, in an upward fashion without crossings; see \Cref{fig:block-replacement}.  Note that the the vertices in~$X_3$ and in~$X_4$ lie on pairwise distinct levels, and the respective block replacement edge crosses all such levels.  Thus, in the modified drawing, the sum of their spans exceeds the number of removed upper $uv$-ears by at least~$2$.  Observe further that the spans of all other edges do not change.  
  We treat the lower $uv$-ears similarly, if there are more than six of them.  Applying this modification to each pair~$u,v$ that has more than six upper or lower $uv$-ears, yields an $s$-good drawing of~$G'$.
\end{proof}

\begin{figure}
    \centering
    \begin{subfigure}[b]{.3\textwidth}
    \centering
    \includegraphics[page=5]{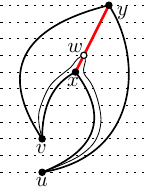}
    \subcaption{}\label{fig:ear-reinsertion}
    \end{subfigure}
    \begin{subfigure}[b]{.3\textwidth}
    \centering
    \includegraphics[page=3]{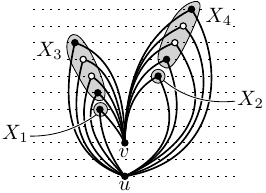}
    \subcaption{}\label{fig:ear-partition}
    \end{subfigure}
    \begin{subfigure}[b]{.3\textwidth}
    \centering
    \includegraphics[page=4]{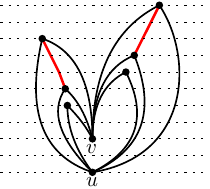}
    \subcaption{}\label{fig:block-replacement}
    \end{subfigure}
    \caption{(a) Reinsertion of a removed upper ear on an upper block replacement edge. (b) Consecutive sets of upper clockwise and counterclockwise $uv$-ears. (c) Replacement of all but six of the $uv$-ears by two upper block replacement edges (drawn red).}
\end{figure}

We are now ready to prove our main result.

\begin{theorem}
  \label{thm:fpt-vc}
    Computing the span of a DAG~$G$ is FPT with respect to $\vc(G)$.
\end{theorem}

\begin{proof}
Given our planar input graph~$G=(V,E)$ together with a vertex cover~$M$ of size~$k$, we first construct the graph~$G'=(V',E')$ by (i) removing all but one upper and one lower degree-$1$-neighbor for each vertex in~$M$ (at most $2k$ vertices in total), (ii) removing all but one transversal degree-$2$-vertex for each pair~$u,v \in M$ (at most $3k$ vertices in total), and (iii) reducing for each pair~$u,v \in M$ the number of upper and lower~$uv$-ears to at most~$6$ by introducing upper and lower block replacement edges (at most~$12 \cdot 3k$ vertices).  Together with the $k$ vertices from~$M$ and the at most~$2k$ vertices of~$G-M$ that have degree at least~$3$, we thus get that~$G'$ has at most~$44k$ vertices.

By \Cref{lem:ears-replaced}, $G$ has an upward-planar layered drawing with span~$s$ if and only if $G'$ has an $s$-good drawing.  To test the latter, we exploit the fact that the size of~$G'$ is bounded.  As the number of ways in which $G'$ can be augmented to be $st$-planar as well as the number of ways the augmented graph can be embedded planarly is bounded by a function of~$k$, we can branch over all choices of an augmentation edges~$X$ and an embedding~$\mathcal E$ in FPT time.  
It thus suffices to test whether~$G'+X$
admits an $s$-good drawing with the given embedding~$\mathcal E$, where the span of the edges in~$X$ may be arbitrarily large.  Recall from \Cref{sec:algorithms} that testing whether~$G'+X$ has an upward-planar layered drawing with span~$s$ with embedding~$\mathcal E$, where the edges in~$X$ may have arbitrarily large span, can be modeled as a network flow problem, and thus also as an integer linear program (ILP) $I$ whose number of variables and clauses is bounded linearly in~$k$.  It is now straightforward to add the missing conditions for an $s$-good drawing (i.e., the lower bounds on the sum of spans for the pairs of block replacement edges) to this ILP to obtain an ILP $I'$ whose number of variables and clauses is still bounded linearly in~$k$ and that has a solution if and only if~$G'+X$ has an $s$-good drawing with embedding~$\mathcal E$, where edges in~$X$ may have arbitrary span.  This can be tested in FPT time with respect to~$k$, since ILPs can be solved in FPT time with respect to the number of variables~\cite{lenstra-ipfnv-83}.
\end{proof}

\section{Conclusion}\label{sec:conclusion}

In this paper, we studied the span of directed graphs. Our research opens up many questions. From a combinatorial perspective, \cref{prop:binaryTree,cor:Tk} rule out the possibility of an upper bound for the span of $n$-vertex degree-$d$ directed trees as a function of the length $\ell$ of the longest directed path only, or of $n+d$ with a sub-linear dependence on~$n$. While our $O(d^{\ell-1})$ bound (\cref{th:ub-path-degree}) is almost tight (\cref{cor:Tk}), the $O(\ell\cdot n^{0.695})$ bound of \cref{th:ub-longest-path} leaves room for improvement: Can the exponent of~$n$ be lowered via a tree decomposition better than the one obtained by repeatedly splitting the tree along greedy paths? Can a sub-linear function of~$n+\ell$ be achieved? From an algorithmic perspective, \cref{thm:nph-single-source,thm:nph-trees} rule out the possibility of solving $k$-\textsc{span upward planarity} efficiently (unless P=NP), even for $k=2$ and for directed trees or single-source graphs. However, the problem might be tractable for instances with a fixed planar embedding, e.g., for triconnected graphs. Also, it would be interesting to provide an FPT algorithm parameterized by the treedepth, thus strengthening our results in %
\cref{thm:fpt-vc,thm:fpt-para+span}.

\bibliography{references}
\end{document}